\documentclass[11pt,prd,aps,amsfonts,eqsecnum,superscriptaddress,nofootinbib,longbibliography,notitlepage]{revtex4-2}

\usepackage{lmodern}    
\usepackage{microtype}  
\usepackage[english]{babel}
\usepackage{setspace}

\usepackage{fancyhdr}

\usepackage{caption}
\usepackage{graphicx}
\usepackage{xcolor}
\usepackage{subcaption}
\usepackage{rotating}
\usepackage{amsmath, amssymb, graphics, amsthm}
\usepackage{physics}
\usepackage{verbatim}
\usepackage{tikz}
\usepackage{tikz-cd}
\usepackage{circuitikz}
\usetikzlibrary{patterns}
\usetikzlibrary{through,hobby}
\usetikzlibrary{decorations.markings}
\usetikzlibrary{fadings,shadings}
\usetikzlibrary{plotmarks}
\usetikzlibrary{positioning}
\usetikzlibrary{decorations,arrows}
\usetikzlibrary{decorations.pathreplacing}
\usetikzlibrary{chains, scopes, positioning, backgrounds, shapes, fit, shadows, calc, arrows.meta, decorations.pathreplacing}
\usepackage{blkarray}
\usepackage{svg}
\usepackage{mathrsfs}
\usepackage{algorithm}
\usepackage{algpseudocode}

\usepackage{xcolor}
\definecolor{mycitecolor}{rgb}{0.0, 0.45, 0.85}   

\newcommand{\cse}[1]{\par\medskip\noindent\emph{#1}}
\usepackage[
  colorlinks=true,
  linkcolor=mycitecolor,
  citecolor=mycitecolor,
  urlcolor=mycitecolor,
  hyperindex=true,
  linktocpage=true
]{hyperref}

\usepackage[capitalise,compress]{cleveref}
\usepackage{tcolorbox}
\usepackage{graphicx}
\usepackage{xcolor}
\usepackage{subcaption}
\usepackage{rotating}
\usepackage{amsmath,amssymb,graphics,amsthm}
\usepackage{physics }
\usepackage{verbatim}
\usepackage{circuitikz}
\usepackage{tikz}
\usepackage{esint}
\usepackage{blkarray}
\usepackage{array}
\usepackage{svg}
\usepackage{ytableau}
\usepackage{cancel}

\usepackage{enumitem}

\newlist{thmcases}{enumerate}{1}
\setlist[thmcases]{
    label=\textbf{Case \arabic*:}, 
    leftmargin=*,                  
    topsep=0.5em,                  
    parsep=0.5em                   
}

\newtcolorbox{shadedtheorem}{
  colback=gray!15,    
  colframe=white,     
  boxrule=0pt,        
  arc=0pt,            
  left=5pt,           
  right=5pt,
  top=5pt,
  bottom=5pt
}
\usetikzlibrary{calc}

\newtheoremstyle{upright-hang}
  {6pt}   
  {6pt}   
  {\normalfont\hangindent=1.5em\hangafter=1\relax} 
  {0pt}   
  {\bfseries} 
  {.}     
  {0.5em} 
  {}      

\theoremstyle{upright-hang}
\newtheorem{thm}{Theorem}
\numberwithin{thm}{section}
\newtheorem{cor}[thm]{Corollary}

\newtheorem{lem}[thm]{Lemma}

\newtheorem{defn}[thm]{Definition}

\renewcommand{\thesection}{\arabic{section}}
\renewcommand{\thesubsection}{\thesection.\arabic{subsection}}
\renewcommand{\thesubsubsection}{\thesubsection.\arabic{subsubsection}}

\makeatletter
\renewcommand{\p@subsection}{}
\renewcommand{\p@subsubsection}{}
\makeatother

\makeatletter
\fancypagestyle{plain}{%
  \fancyhf{}%
  \fancyfoot[C]{\thepage}%
}
\fancypagestyle{revtexfooter}{%
  \fancyhf{}%
  \fancyfoot[C]{\thepage}%
}
\makeatother
\AtBeginDocument{\setstretch{1.25}\pagestyle{revtexfooter}\raggedbottom}
\makeatletter
\renewcommand\frontmatter@title@format{\LARGE\bfseries\centering\parskip\z@skip}

\def\frontmatter@affiliationfont{\normalfont\selectfont}
\newcommand{\AndySectionStar}[1]{%
  \par\addpenalty\@secpenalty
  \addvspace{3.5ex \@plus 1ex \@minus .2ex}%
  \noindent{\normalfont\Large\bfseries #1\par}%
  \nobreak\vspace{2.3ex \@plus .2ex}%
  \@afterheading}
\newcommand{\AndySubsectionStar}[1]{%
  \par\addpenalty\@secpenalty
  \addvspace{3.25ex \@plus 1ex \@minus .2ex}%
  \noindent{\normalfont\large\bfseries #1\par}%
  \nobreak\vspace{1.5ex \@plus .2ex}%
  \@afterheading}
\newcommand{\AndySubsubsectionStar}[1]{%
  \par\addpenalty\@secpenalty
  \addvspace{3.25ex \@plus 1ex \@minus .2ex}%
  \noindent{\normalfont\normalsize\bfseries #1\par}%
  \nobreak\vspace{1.5ex \@plus .2ex}%
  \@afterheading}
\renewcommand\@seccntformat[1]{\csname the#1\endcsname\quad}
\def\@hangfrom@section#1#2#3{\@hangfrom{#1#2}#3}
\def\@hangfroms@section#1#2{#1#2}
\renewcommand\section{\@ifstar{\AndySectionStar}{\@startsection {section}{1}{\z@}%
  {-3.5ex \@plus -1ex \@minus -.2ex}%
  {2.3ex \@plus .2ex}%
  {\normalfont\Large\bfseries}}}
\renewcommand\subsection{\@ifstar{\AndySubsectionStar}{\@startsection{subsection}{2}{-\parindent}%
  {-3.25ex\@plus -1ex \@minus -.2ex}%
  {1.5ex \@plus .2ex}%
  {\normalfont\large\bfseries}}}
\renewcommand\subsubsection{\@ifstar{\AndySubsubsectionStar}{\@startsection{subsubsection}{3}{-\parindent}%
  {-3.25ex\@plus -1ex \@minus -.2ex}%
  {1.5ex \@plus .2ex}%
  {\normalfont\normalsize\bfseries}}}
\renewcommand*\l@section[2]{%
  \ifnum \c@tocdepth >\z@
    \addpenalty\@secpenalty
    \addvspace{0.75em \@plus\p@}%
    \begingroup
      \parindent\z@
      \rightskip\@pnumwidth
      \parfillskip-\@pnumwidth
      \leavevmode\bfseries
      \setlength\@tempdima{2.3em}%
      \advance\leftskip\@tempdima
      \hskip-\leftskip
      #1\nobreak\hfil\nobreak\hb@xt@\@pnumwidth{\hss #2}\par
    \endgroup
  \fi}
\renewcommand*\l@subsection{\@dottedtocline{2}{1.5em}{3.2em}}
\renewcommand*\l@subsubsection{\@dottedtocline{3}{3.8em}{4.1em}}
\makeatother

\usepackage{mathtools}
\tikzstyle{densely dashed}= [dash pattern=on 4pt off 3pt]

\newcommand{\SON}[1]{\mathrm{SO}(#1)}
\newcommand{\SUN}[1]{\mathrm{SU}(#1)}

\usetikzlibrary{matrix, calc}
\tikzset{
  ycell/.style = {
    draw=black,
    line width=0.4pt
  },
  ydel/.style = {
    draw=red!80!black,
    line width=0.9pt
  },
  yadd/.style = {
    draw=green!55!black,
    line width=1.3pt
  }
}
\newcommand{\Ydeleted}[2]{%
  \pgfmathsetmacro{\xleft}{#1-1}
  \pgfmathsetmacro{\xright}{#1}
  \pgfmathsetmacro{\ytop}{1-#2}
  \pgfmathsetmacro{\ybottom}{-#2}

  \foreach \t in {0.08,0.43,0.78}{%
    \draw[ydel]
      ({\xleft+\t},{\ytop}) --
      ({\xleft+\t+0.14},{\ytop});
    \draw[ydel]
      ({\xleft+\t},{\ybottom}) --
      ({\xleft+\t+0.14},{\ybottom});
    \draw[ydel]
      ({\xleft},{\ybottom+\t}) --
      ({\xleft},{\ybottom+\t+0.14});
    \draw[ydel]
      ({\xright},{\ybottom+\t}) --
      ({\xright},{\ybottom+\t+0.14});
  }%
  \foreach \x/\sx in {\xleft/1,\xright/-1}{%
    \foreach \y/\sy in {\ytop/-1,\ybottom/1}{%
      \draw[ydel, solid, line join=miter]
        ({\x+\sx*0.08},\y) --
        (\x,\y) --
        (\x,{\y+\sy*0.08});
    }%
  }%
}
\usepackage{etoolbox}

\newcommand{\Ycell}[3]{%
  \ifstrequal{#3}{ydel}{%
    \Ydeleted{#1}{#2}%
  }{%
    \draw[#3] (#1-1,1-#2) rectangle (#1,-#2);
  }%
}
\newcommand{\Yx}[2]{%
  \draw[red!80!black,line width=0.9pt] (#1-1,1-#2)--(#1,-#2);
  \draw[red!80!black,line width=0.9pt] (#1-1,-#2)--(#1,1-#2);}
\newcommand{\Ydiagram}[1]{%
  \foreach \len [count=\r] in {#1}{%
    \foreach \c in {1,...,\len}{\Ycell{\c}{\r}{ycell}}}}
\newcommand{\dg}[1]{\vcenter{\hbox{\begin{tikzpicture}[x=0.55cm,y=0.55cm]#1\end{tikzpicture}}}}
 
\newcommand{\Rmu}{\dg{\Ydiagram{2,2,2,1}}}
\newcommand{\PairAleft}{\dg{\Ydiagram{2,2,1,1}\Ycell{2}{3}{ydel}\Yx{2}{3}}}
\newcommand{\PairAright}{\dg{\Ydiagram{2,2,2,2}\Ycell{2}{4}{yadd}}}
\newcommand{\PairBleft}{\dg{\Ydiagram{2,2,2}\Ycell{1}{4}{ydel}\Yx{1}{4}}}
\newcommand{\PairBright}{\dg{\Ydiagram{2,2,2,1,1}\Ycell{1}{5}{yadd}}}

\usepackage{dsfont}
\usepackage{booktabs}   
\usepackage{tabularray}
\UseTblrLibrary{booktabs,counter}
\newcounter{assump}
\renewcommand{\theassump}{A\arabic{assump}}
\newcommand{\assumption}[1]{\refstepcounter{assump}\textbf{\theassump}\label{#1}}

\allowdisplaybreaks

\begin{document}

\title{On Higher-Form Fracton Hydrodynamics with Momentum Conservation}

\author{Ben Bobell}
\email{ben.bobell@colorado.edu}

\affiliation{Department of Physics and Center for Theory of Quantum Matter, University of Colorado, Boulder, CO 80309, USA}

\author{Rahul Nandkishore}
\affiliation{Department of Physics and Center for Theory of Quantum Matter, University of Colorado, Boulder, CO 80309, USA}

\author{Andrew Lucas}
\affiliation{Department of Physics and Center for Theory of Quantum Matter, University of Colorado, Boulder, CO 80309, USA}

\begin{abstract}
   Momentum conservation is common in fluids, yet it may be incompatible with the hydrodynamics of theories with exotic conservation laws (such as fracton hydrodynamics) unless symmetries are spontaneously broken. Using the Martin-Siggia-Rose (MSR) effective field theory formalism, we develop a representation-theoretic framework for the hydrodynamics of translationally invariant systems with exotic conservation laws without spontaneous symmetry breaking. Under mild assumptions on the form of hydrodynamics, we show that such theories must satisfy two representation theoretic consistency conditions and a generalized Gauss's law constraint. Higher-form symmetries satisfy these conditions and we derive their ideal hydrodynamics; in three dimensions, the one-form case produces ideal non-relativistic magnetohydrodynamics. Finally, we prove that no other theories are compatible with momentum conservation in $\SON{3}$ nor $\SON{d}$, for sufficiently large odd $d$.
\end{abstract}

\maketitle

\tableofcontents

\section{Introduction}
Kinetic constraints in both quantum and classical systems can dramatically slow thermalization. In the simplest cases, these kinetic constraints are imposed via generalized global symmetries \cite{Gaiotto_2015, Pretko_2016, Pretko_2017,  PPN, Seiberg_2020.1, Seiberg_2020.2, Gorantla_2020, Seiberg_2021, Rudelius_2021, Burnell2022, You_2022, Gorantla_2022,  Sala2022}, whose dynamics can sometimes be understood using \textit{hydrodynamics}--an effective field theory for the near-equilibrium dynamics of local, fluctuating dynamical systems with a conserved charge. One by now well-studied class of such theories involves systems with global charge and dipole (or multipole) conservation, and is known as \textit{fracton hydrodynamics}  \cite{Iaconis_2019, Gromov_2020, Morningstar_2020, Feldmeier_2020, Zhang2020, Feldmeier_2021, Grosvenor2021, Iaconis_2021, Hart_2022, Osborne_2022, Guo_2022, Glorioso_2022, Burchards_2022, Armas_2024, Godkowski_2024, Godkowski_2026}.

\emph{Higher-form symmetries} are another type of generalized symmetry, which act on $p$-dimensional charged objects through a $(d-p)$-dimensional surface in $d$ spatial dimensions \cite{Grozdanov_2019, lake2018higherformsymmetriesspontaneoussymmetry, Hofman_2019, Delacrataz_2020, Gomes_2023}. Textbook magnetohydrodynamics may be reformulated as a hydrodynamic theory for a magnetic one-form symmetry in three spatial dimensions, in which the magnetic flux through an arbitrary two-dimensional surface must be conserved \cite{Grozdanov_2017, glorioso2018effectivefieldtheorymagnetohydrodynamics, Armas_2018, Armas_2019, Armas_2020}. Such exotic theories were generically formalized via group theory for arbitrary groups $G$, with charge densities $\rho_a$ transforming in an irreducible representation $R_a$, in Refs.~\cite{fractonMHD,Qi_2023}\footnote{Throughout this paper, we use \emph{$p$-form symmetry} to refer to a theory whose conserved charge density $\rho_a$ is an antisymmetric rank-$p$ tensor and satisfies a generalized Gauss's law constraint, so that its flux through a $(d-p)$-dimensional surface is conserved. } . This formalism allowed for the description of theories with gauged conservation laws on charge and dipole (or multipole) moments. The symmetry/group theoretic approach has a definite advantage: the hydrodynamic modes depend only on properties of the symmetry itself, and thus avoid any particular microscopic assumptions often needed in textbook derivations.

Source-free magnetohydrodynamics also supports momentum conservation without spontaneously breaking any symmetries. In contrast, the global charge and dipole conserving model was shown in Refs.~\cite{Glorioso_2022, Glorioso_2023} to exhibit spontaneous symmetry breaking, since momentum and the dipole moment generically do not commute. But what if the symmetries are gauged? It is natural to ask: what generalizations of fracton hydrodynamics \textit{are} compatible with momentum conservation without additional spontaneous symmetry breaking? We adopt the abstract framework in Ref.~\cite{Qi_2023}, but allow generic charge densities $\rho_a$ supported on a potentially finite volume $V$. 

We first take a constructive approach, deriving two consistency conditions (dubbed C1 and C2) that are necessary for a theory to conserve momentum. Momentum conservation will impose a generalized Gauss's law constraint (of the form $\Gamma \partial^m \rho_a=0$, for a tensor $\Gamma$), which in turn forces any nonzero charge to the boundary $\partial V$; nontrivial finite charges exist only if $\rho_a $ is supported on a finite volume $V$. We argue that only theories with a single-derivative Gauss's law constraint ($m=1$) are compatible with momentum conservation, and thus one cannot formulate a consistent gauged fracton magnetohydrodynamics with momentum conservation, at least in theories with purely scalar multipole charges. Among these $m=1$ theories with $\SON{3}$ symmetry, our two conditions leave higher-form symmetries as the only ones compatible with momentum conservation.  We conjecture that this result holds more generally in higher dimensions, i.e., for the group $\SON{d}$, and present some partial results towards this end. Although a lot of representation theoretic machinery goes into these results, we also provide a more physical argument for why non-trivial conservation laws in a momentum-conserving theory must contain at least a one-form symmetry, without breaking the structure of ideal hydrodynamics.
Momentum conservation thus does not come for free: outside of higher-form symmetries, it then seems to force an emergent Goldstone boson. 

Our paper is structured as follows: in Section \ref{Sec:scalarMSR} we provide an overview of the Martin-Siggia-Rose (MSR) effective field theory for scalar charges \cite{MSROG, Guo_2022}, before expanding it to generic conserved charges in Section \ref{Sec:Genericcharge}. As an example of our framework, we derive non-relativistic magnetohydrodynamics at ideal order. In Section \ref{Sec:consistencyconditions}, we derive two consistency conditions necessary for a theory to be compatible with momentum conservation. In Section \ref{Sec:PossibleRI}, after providing a pedagogical review of Young tableaux, we determine the possible irreducible representations compatible with consistency condition C1. We conclude in Sections \ref{Sec:Idealsolutions} and \ref{Sec:likelynothingelse} by showing that higher-form symmetries are compatible with momentum conservation in $\SON{d}$, and that essentially no other theories are.

\section{Review of MSR Formalism}
\label{Sec:scalarMSR}
Before describing exotic conserved charges, we review the Martin-Siggia-Rose (MSR) effective field theory framework for ordinary conserved densities \cite{MSROG}. Our review mirrors that of Ref.~\cite{Qi_2023}. These effective field theories describe stochastic fluctuating systems near their steady state with local dynamics. We consider a conserved scalar density $\rho(x,t)$ in $d$ spatial dimensions, defined on some possibly infinite volume $V\subseteq \mathbb{R}^d $. The conserved charge associated to such a density and volume will be $Q(t;V)\coloneq \int_V d^dx \rho(x,t)$. For $Q$ to truly be a conserved charge, we must enforce $dQ/dt=0$. Such charge-conserving dynamics are encoded via the continuity equation
\begin{equation}
    \partial_t\rho +\partial_iJ_i=0,
\end{equation}
where $J_i$ is the current density whose explicit form is yet to be determined and repeated indices will henceforth imply summation unless explicitly stated otherwise. Since we define our theory on an arbitrary volume, we demand that no charge leaks out of the volume $V$, i.e. $J_i\hat{n}_i|_{\partial V}=0$, where $\hat{n}_i$ denotes the normal vector to the boundary $\partial V$. Such a constraint is necessary because
\begin{equation}
    0 = \frac{d Q}{d t} = \int_V d^dx \partial_t\rho =  -\int_V d^dx \partial_i J_i = -\int_{\partial V} \hat{n}_iJ_i. 
\end{equation}
Demanding this for generic $V$ requires the integrand to vanish, giving $J_i\hat{n}_i|_{\partial V}=0$. Our goal is to encode these dynamics as a classical effective field theory. We write the action $S$ in terms of $\rho(x,t)$ and the conjugate \textit{noise field} $\pi(x,t)$, which encodes the random fluctuations in the microscopic dynamics and acts like a conjugate momentum for $\rho$:
\begin{equation}
    S= \int dt d^dx \left(\pi \partial_t\rho-H(\pi, \rho)\right).
\end{equation}
Here, $H(\pi,\rho)$ is the MSR effective Hamiltonian, whose explicit form is to be determined. We demand that $H$ contains no $\pi$-independent terms, $H(0,\rho)=0$ (which is related to demanding that the stochastic dynamics for $\rho$ have a normalized probability distribution). We choose this action $S$ because taking variational derivatives with respect to $\pi$ and setting $\pi$ to zero will generate the dynamics of $\rho$:
\begin{equation}
    0=\frac{\delta S}{\delta \pi}\Big|_{\pi =0}= \partial_t\rho - \frac{\partial H}{\partial  \pi}\Big|_{\pi =0}.
\end{equation}
Comparing with the continuity equation and integrating by parts to move the derivative off $J_i$, we obtain $H(\pi, \rho)=\partial_i \pi J_i $ (up to a boundary term that vanishes via our boundary condition on $J_i$). The principle of the MSR formalism is thus to write the most general $H$ which is consistent with the symmetries of our theory and match terms to find $J_i$. Such a Hamiltonian generically has infinitely many terms, corresponding to arbitrarily many spatial derivatives on $\pi$. At \textit{ideal order}, $J_i$ is independent of derivatives of $\pi$, while at $n$-th order, $J_i$ is written in terms of at most $n$ derivatives. 

We constrain $H$ by imposing the symmetries of the theory as invariance under shifts of $\pi$. In the most minimal theory, $H$ is invariant under spatially-independent shifts $c(t)$ of the noise field, $\pi \to \pi +c(t)$. Under such a shift, the action changes by
\begin{equation}
    S \to S + \int dt  c(t) \partial_t Q.
\end{equation}
Invariance of the action under this shift is precisely the statement that $\partial_tQ=0 $ because $c(t)$ is arbitrary. We often additionally demand time-reversal symmetry (invariance under $t\to -t$). To do so, we must first postulate the existence of a stationary state of the form
\begin{equation}
    P_s[\rho] = e^{-\Phi[\rho]},
\end{equation}
 where $\Phi[\rho]$ is a functional of the density and acts like a thermodynamic potential. We can define a conjugate chemical potential
 \begin{equation}
     \mu(x)\coloneq \frac{\delta \Phi}{\delta \rho(x)}
 \end{equation}
 and, though we shall not use it here, the thermodynamic pressure $P$ as the Legendre transform of $\Phi$, satisfying $\rho\equiv \delta P/\delta \mu(x) $. It was shown in Ref.~\cite{Guo_2022} that, in the weak noise limit, time-reversal symmetry corresponds to $t\to -t$ together with the transformation
 \begin{equation}
     \pi \to -\pi +i\mu  .
 \end{equation}
 Finally, we assume that statistical fluctuations are bounded:
 \begin{equation}
     \mathrm{Im}(H)\leq 0.
 \end{equation}
Together, these constraints fix the MSR Hamiltonian $H$. The $\pi$ shift invariance forces each term to depend on a spatial derivative of $\pi$, time reversal symmetry forces these terms to be accompanied by a $\pi-i\mu$ term, and $\mathrm{Im}(H)\leq 0$ forces the sign of the transport coefficient. At first order in one-spatial dimension ($d=1$), $H$ shall be
 \begin{equation}
     H=  -i\sigma(\mu) \partial_x\pi\partial_x(\pi-i\mu),
 \end{equation}
 where $\sigma(\mu)$ is a function of $\mu$ that contains no spatial derivatives. It is natural to rephrase the dynamics in terms of the chemical potential $\mu$: using $\partial_t\rho = \chi \partial_t\mu  $ (with the thermodynamic susceptibility $ \chi \coloneq \partial \rho/\partial \mu $), taking a variational derivative with respect to $\pi $, and setting $\pi =0$ generates the equation of motion
 \begin{equation}
     \partial_t \mu -\frac{1}{\chi }\partial_x(\sigma \partial_x\mu)=0,
 \end{equation}
 which implicitly defines its current. In the following section, we shall generalize this construction to densities transforming in exotic representations of $\SON{d}$ and impose momentum conservation.

\section{General Framework for Exotic Conserved Charges }\label{Sec:Genericcharge}
We now generalize our discussion to more exotic and generic charge densities. Before we develop our framework, we enumerate the assumptions we make about the form that ideal hydrodynamics can take in Table~\ref{tab:assumptions}. The consistency conditions derived in Section~\ref{Sec:consistencyconditions} are consistency conditions with respect to these assumptions.

\begin{table}[t]
\centering
\small
\begin{tblr}{colspec={@{} l Q[l,wd=0.18\linewidth] Q[j,wd=0.66\linewidth] @{}}, rowsep=2pt}
\toprule
 & \textbf{Keyword} & \textbf{Assumption Statement} \\
\midrule
\assumption{as:locality} & Locality &
The stress tensor $\tau_{ji}$ and current $J_{ai}$ are local functions of
the hydrodynamic variables, as is required for the validity of hydrodynamics. \\
\assumption{as:symmetry} & Symmetries &
Translational and $\SON{d}$ symmetry for $d\geq 3$ (although our framework applies to general finite-dimensional groups $G$). Theorem~\ref{thm:pform} additionally requires that $d$ is odd (we believe this can be relaxed) and Section~\ref{sec:RIchoice} assumes $d\neq 4$. We do not assume boost invariance. \\
\assumption{as:msr} & MSR Hydrodynamics &
Time-reversal symmetry with a stationary state $P_s=e^{-\Phi}$. \\
\assumption{as:exact} & Exactness &
Every solution of $\Gamma_{ia\alpha}\partial_i f_a=0$ may be written as $f_a=\Gamma_{j_1\cdots j_maI}\partial_{j_1}\cdots \partial_{j_m}g_I$, for some $g_I$ transforming in an irreducible representation $R_I$. We are unaware of any theories that violate this exactness assumption, but any such theory lies outside the scope of our framework. \\
\assumption{as:entropy} & Entropy Current  &
The entropy current takes the form $s_i=Pv_i+\cdots$ (equivalently, $J_{ai}=\rho_av_i+\cdots$). We are unaware of any momentum-conserving theories that violate this form.  \\
\assumption{as:boundary} & Boundary Conditions &
 No current flows out of our volume $V$ (so $\hat{n}_iJ_{ai}|_{\partial V}=0$), derivatives of the density $\rho_a$ behave nicely on the boundary ($\hat{n}_j \Gamma_{jaI} \partial_k \rho_a|_{\partial V}=0 $), and we work in the sector with no dynamically frozen background ($c_I=0$, see Section~\ref{sec:generalizedgausslawsec}). \\
\assumption{as:generic} & Genericness &
All conditions are required to hold for arbitrary densities allowed by symmetries and constraints. For instance, the condition $\Gamma_{ia\alpha}\partial_if_a=0$ follows from requiring $ \int_V d^dx (\partial_if_a)\Gamma_{ia\alpha}J_{\alpha}=0$ for arbitrary $J_\alpha$.  \\
\bottomrule
\end{tblr}
\caption{Assumptions under which C1, C2, and Theorem~\ref{thm:pform} are derived.}
\label{tab:assumptions}
\end{table}

Consider a system in $d$ spatial dimensions whose dynamics remain invariant under the group $G$, characterized by a conserved charge density $\rho_a$ and a current density $J_{\alpha}$. To be as general as possible, let $Q_a= \int_V d^dx \rho_a$, for some large, possibly infinite volume $V\subseteq \mathbb{R}^d $. The charge density $\rho_{a}$ transforms in the irreducible representation $R_a$, while $J_\alpha$ transforms in a generally reducible representation $ R_\alpha = \bigoplus_k R_{\alpha_k}$. For the remainder of this paper, we will select $G=\SON{d}$ for convenience, although the framework developed here generalizes to any finite-dimensional group $G$ (Assumption~\ref{as:symmetry}). Likewise, we are interested in theories exhibiting momentum conservation, featuring a conserved momentum density $g_i$ that transforms in the $d$-dimensional vector representation $\mathcal{V}$, alongside a corresponding stress tensor $\tau_{ij}$ (Assumption~\ref{as:symmetry}). The continuity equations governing these dynamics are given by
\begin{equation}
\label{eq:continuityeqn}
   \partial_t\rho_a+\partial_i J_{ai} =0, \quad  \partial_t g_j +\partial_i \tau_{ji}=0,  
\end{equation}
where $J_{ai}= \Gamma_{ia\alpha}J_{\alpha}$ and the tensors $\Gamma_{ia\alpha}$ are generalized Clebsch-Gordan coefficients. More generally, a rank-$m$ generalized Clebsch-Gordan coefficient $\Gamma_{i_1\cdots i_m a\alpha}$ is nonzero if $R_a$ appears in the irreducible decomposition of $( \mathcal{V}^{\otimes m} )\otimes R_\alpha $. 
For $Q_a$ to be a well-defined charge, we require that current perpendicular to $\partial V$ must identically vanish: $\hat{n}_i J_{ai}|_{\partial V}=0$, where $\hat{n}_i$ is the unit vector perpendicular to $\partial V$ (Assumption~\ref{as:boundary}). Such a condition enforces that no charge can leave the volume $V$. We wish to study conserved quantities of the form
\begin{equation}
    Q[f_a] \coloneq \int_V d^dx f_a\rho_a,
\end{equation}
where $f_a$ are functions of $x$.  Clearly, $f_a$ may be taken to be a constant. We may also consider whether it is possible to take $x$-dependent $f_a$. For $Q[f_a]$ to be a conserved quantity, its time derivative must vanish:
\begin{equation}
   0= \frac{d Q[f_a]}{dt} =  \int_V d^dx f_a\partial_t\rho_a =  -\int_V d^dx f_a\Gamma_{ia\alpha}\partial_iJ_{\alpha}=  \int_V d^dx (\Gamma_{ia\alpha}\partial_if_a)J_{\alpha}.
\end{equation}
In the final step, we integrated by parts, utilizing the fact that $\hat{n}_i\Gamma_{ia\alpha}J_\alpha|_{\partial V}=0$ to drop the boundary term. Since this must hold for arbitrary currents (Assumption~\ref{as:generic}), we require \begin{equation}
    \Gamma_{ia\alpha}\partial_i f_a=0.
\end{equation}
We can now encode the above constraints in terms of the MSR effective field theory framework. We introduce the conjugate noise field $ \pi_a$ and chemical potential $\mu_a$, corresponding to our charge density $\rho_a$, alongside the conjugate noise field $\sigma_i$ and chemical potential $v_i$ (the fluid velocity) corresponding to our momentum density $g_i$. The MSR action is given by
\begin{equation}
    S= \int dt d^dx \left[\pi_a\partial_t\rho_a +\sigma_i\partial_tg_i-H(\pi_a,\sigma_i,\rho_a,g_i)\right],
\end{equation}
where the Hamiltonian $H(\pi_a,\sigma_i,\rho_a,g_i)$ satisfies $H(0,0,\rho_a,g_i)=0$ and the time-reversal constraint:
\begin{equation}
    \label{eq:TRS}
    H(i\mu_a,iv_i,\rho_a,g_i) = i\partial_i s_i.
\end{equation}
Here, $s_i$ is proportional to the entropy current density. To encode the equations of motion given in Eq.~\eqref{eq:continuityeqn} and ensure invariance under $\pi_a\to \pi_a+f_a$, we write $H$ as 
\begin{equation}
    H(\pi_a,\sigma_i,\rho_a,g_i)= \partial_i \pi_a  J_{ai} + \partial_i \sigma_j \tau_{ji},
\end{equation}
where we shall henceforth adopt an index convention in which \textbf{each generalized index runs over an orthonormal basis of its own representation and separate indices are summed independently.} For instance, an antisymmetric $p$-form index $a$ takes the form $a=(a_1,\ldots ,a_p)$, where $a_1<\ldots<a_p$ (i.e., its components are \textit{strictly ordered and unique}). If $G=\SON{3}$ and $a$ is a 3-form index, the only allowed value is $a_1a_2a_3=123$, so
\begin{equation}
    \epsilon_a\epsilon_a = \sum_{a_1<a_2<a_3 }\epsilon_a^2 =\epsilon_{123}^2 =1.
\end{equation}
However, if $i$, $j$, $k$ are three separate vector indices, each is summed independently:
\begin{equation}
    \epsilon_{ijk} \epsilon_{ijk}= \sum_{ijk}\epsilon_{ijk}^2 =6. 
\end{equation}
The value of an index contraction thus depends on how indices are grouped into generalized indices; the two results differ by $3!$, the number of orderings of $i$, $j$, and $k$. To avoid confusion, each index letter ($i$, $j$, $a$, etc.) denotes a single generalized index, and all index contractions are written in terms of these labels, unless explicitly stated otherwise. This convention avoids cumbersome combinatorial prefactors and makes $\partial\rho_a/\partial \rho_b=\delta_{ab}$ (i.e., the identity on $R_a$).

With our framework established, we now turn our attention to determining the constitutive relations. In Section~\ref{Sec:obstructionansatz}, we show that momentum conservation must impose a constraint on the allowed charge densities $\rho_a$. We then demonstrate that in ideal non-relativistic magnetohydrodynamics such a constraint is precisely Gauss's law (Section~\ref{Sec:IdealMHD}), which provides a template for constructing the constraint for arbitrary charges in Section~\ref{Sec:consistencyconditions}.

\subsection{Obstruction to a Simple Ansatz}
\label{Sec:obstructionansatz}

First, we will naively assume that momentum conservation places no constraints on our charge density and observe that this will lead to a contradiction with Assumption~\ref{as:generic}. Inserting our choice for $H$ into Eq.~\eqref{eq:TRS}, we arrive at
\begin{equation}
     (\partial_i \mu_a ) J_{ai} + (\partial_i v_j )\tau_{ji} =\partial_i s_i = \frac{\partial s_i}{\partial \mu_a}\partial_i \mu_a + \frac{\partial s_i}{\partial v_j}\partial_i v_j ,
\end{equation}
where the last step follows from the chain rule.  Equating the coefficients of the gradients $\partial_i v_j$ and $\partial_i \mu_a$, we extract the constitutive relations:
\begin{equation}
    \tau_{ji} = \frac{\partial s_i}{\partial v_j}, \quad J_{ai} = \frac{\partial s_i}{\partial \mu_a}.
\end{equation}
To annihilate the current, we define the complementary representation $R_{\tilde{\alpha}}$. These are representations present in the irrep decomposition for $\mathcal{V}\otimes  R_{a} $ that are excluded from the current's representations $\oplus_{k}R_{\alpha_k}$, meaning $ (\oplus_{k}R_{{\alpha}_k}) \cap (\oplus_k R_{\tilde{\alpha}_k})=\emptyset $. Because $J_{ai}= \Gamma_{ia\alpha}J_\alpha$, projecting into $R_{\tilde{\alpha}}$ by applying the Clebsch-Gordan coefficients $\Gamma_{ia\tilde{\alpha}}$ to $J_{ai}$ must entirely annihilate $J_{ai}$:
\begin{equation}
    \label{eq:killing}
   0= \Gamma_{ia\tilde{\alpha}} J_{ai}=\Gamma_{ia\tilde{\alpha}} \frac{\partial s_i}{\partial \mu_a}.
\end{equation}
To proceed, we introduce the insightful ansatz $s_i =Pv_i- A_{bi}\rho_b$, where the unknown coefficient $A_{bi}$ can depend on $\mu_a$ and $v_i$ and $P$ is the thermodynamic pressure (Assumption~\ref{as:entropy}). This ansatz is motivated by standard hydrodynamics, where the entropy current contains a $Pv_i$ term. Since $\rho_a \equiv \partial P/\partial \mu_a$, this term contributes $\rho_av_i$ to the current. Plugging this ansatz into Eq.~\eqref{eq:killing} and defining the susceptibility $\chi_{ab}\coloneq \partial \rho_a/ \partial\mu_b$, we obtain
\begin{equation}
    \Gamma_{ia\tilde{\alpha}} \left(\rho_av_i -\frac{\partial A_{bi}}{\partial \mu_a}\rho_b-A_{bi}\chi_{ab}  \right)=0.
\end{equation}
Because the condition must hold for arbitrary densities and potentials (Assumption~\ref{as:generic}), we can separate the terms proportional to the density from those proportional to the susceptibility. Satisfying this condition in general requires, by Assumption~\ref{as:generic}, $A_{bi}=0$ (more precisely, its projection onto $R_{\tilde{\alpha}}$ must vanish, as it is the only term proportional to $\chi_{ab}$), yielding
\begin{equation}
      \Gamma_{ia\tilde{\alpha}}  \rho_av_i=0.
\end{equation}
This condition places a strict physical constraint on $\rho_av_i$: no component of $\rho_av_i$ (which transforms under $\mathcal{V}\otimes R_a $) can lie in the representation $R_{\tilde{\alpha}}$. However, for generic choices of $\rho_a$ and $v_i$, the product $\rho_av_i $ will generically have a non-trivial component in $R_{\tilde{\alpha}}$. Requiring that the overlap must vanish thus provides a contradiction with  Assumption~\ref{as:generic}, demonstrating that the naive assumption that momentum conservation does not provide a constraint on $\rho_a$ cannot generically hold. 

We must therefore explicitly account for the constraint that momentum conservation places upon $\rho_a$. For a constraint to be useful in restricting the current, it must involve a single spatial derivative on the charge density $\rho_a$. This is because only the gradients $\partial_i\mu_a$ and $\partial_i v_j$ enter the constitutive relations in the above derivation; any constraint on $\rho_a$ must match this order to be able to naturally restrict the current. As we will demonstrate in Section \ref{Sec:consistencyconditions}, there exists an irreducible representation $R_I$ such that $\Gamma_{iaI}\partial_i \rho_a=0$.

\subsection{Ideal Non-Relativistic Magnetohydrodynamics}
\label{Sec:IdealMHD}

As a warm-up, we shall solve the case of ideal, non-relativistic magnetohydrodynamics before we consider generic $\rho_a$. We shall work in three spatial dimensions, so our group $G=\SON{3}$, and we let the charge density $\rho_i$ and current $J_k$ transform in the vector representation $\mathcal{V}$. Our continuity equation is
\begin{equation}
    \partial_t \rho_i +\epsilon_{ijk}\partial_j J_k=0,
\end{equation}
which is Faraday's law in textbook Maxwell electromagnetism under the identification of $\rho_i=B_i$ (the magnetic field) and $J_k=E_k$ (the electric field). We study conserved charges of the form
\begin{equation}
    Q[f_i;V] = \int_V d^3x  f_i \rho_i,
\end{equation}
where we demand $\epsilon_{ijk}\partial_j f_i=0$. Since $f_i$ is a curl-free function on a simply connected domain in $\mathbb{R}^3$, it can be written as the gradient of some scalar function $g$: $ f_i=\partial_i g $. Such a constraint produces the standard magnetic one-form symmetry. Let $V' \subset V$ and choose $g $ to be the indicator function of the volume $V'$
\begin{equation}
    g(x;V') = \begin{cases}
        1 & \text{ if } x \in V' \\
        0 & \text{otherwise}
    \end{cases},
\end{equation}
so that $f_i(x)=-\delta(x\in \partial V') \hat{n}_i(x)$, where $\hat{n}_i$ is the unit normal vector. Plugging this into $Q[f_i;V]$ yields the one-form charge
\begin{equation}
    Q[f_i;\partial V'] = \int_{ \partial V'} dS \rho_i \hat{n}_i.
\end{equation}
One-form symmetries thus lie within our class of generalized symmetries.

Momentum conservation further constrains $\rho_i$. Demanding that $\{P_k, Q[f_i;V]\}=0$ for arbitrary $g$ yields 
\begin{equation}
    0=\{P_k, Q[f_i;V]\}= \int_{V} d^3x f_i \partial_k \rho_i =\int_{V} d^3x \partial_i g \partial_k \rho_i = \int_{\partial V} dS \hat{n}_i  g \partial_k \rho_i -  \int_{V} d^3x g \partial_i\partial_k \rho_i. 
\end{equation}
Choosing an infinite family of $g$ that vanish on $\partial V$ (which is always accessible via step functions) causes the bulk integral to vanish for every such choice in $g$, implying that $\partial_k \partial_i\rho_i=0$. A vanishing gradient gives the Gauss's law constraint $\partial_i \rho_i=c(t)$, for $c(t)$ a constant in space but in principle time-dependent. In fact, $c(t)$ must be constant in time. Applying $\partial_i$ to our continuity equation yields
\begin{equation}
   0= \partial_t(\partial_i\rho_i) +\epsilon_{ijk}\partial_i\partial_j J_k = \partial_t(\partial_i\rho_i) = \partial_tc(t).
\end{equation}
Because $\partial_i\rho_i$ is constant in space and time, its dynamics are frozen and we can choose $\partial_i\rho_i=0$ (Assumption~\ref{as:boundary})  without loss of generality (we can always subtract off the dynamically frozen $ cx_i/3 $ part without affecting our hydrodynamics). 

With our Gauss's law constraint $\partial_i\rho_i=0$ chosen, we return to studying the Poisson bracket for a generic $g$. The bulk term now vanishes and the Poisson bracket reduces to the boundary term
\begin{equation}
    \{P_k, Q[f_i;V]\}=\int_{\partial V}dS\hat{n}_i g\partial_k \rho_i  ,
\end{equation}
 which itself must vanish for \textit{arbitrary} $g$. This holds if we either restrict ourselves to a class of theories in which $g|_{\partial V}=0$ or if we impose on $\rho_i$ the boundary condition $\hat{n}_i \partial_k \rho_i|_{\partial V}=0  $ (Assumption \ref{as:boundary}). We reject the former option to allow for nontrivial global charges: selecting $g|_{\partial V}=0$ imposes
\begin{equation}
    Q[f_i] = \int_V d^3x f_i \rho_i =  \int_V d^3x \partial_i g \rho_i = \int_{\partial V} dS\hat{n}_i  g \rho_i - \int_{V}d^3x g \partial_i\rho_i=0,
\end{equation}
which entirely removes nontrivial charges from our theory. We further note that subtracting off our background $ cx_i/3$ is now necessary for momentum conservation. The derived boundary condition $\hat{n}_i \partial_k \rho_i|_{\partial V}=0$ is required  for the Poisson bracket $\{P_k, Q[f_i;V]\}$ to vanish. If $\rho_i$ satisfies this boundary condition, a configuration with the background, $\rho_i'=\rho_i +cx_i/3$ violates it:
\begin{equation}
   \hat{n}_i \partial_k (\rho_i')|_{\partial V}= \hat{n}_i \partial_k (cx_i/3)|_{\partial V} = (c/3)\hat{n}_k \neq 0 \qquad \text{if } c\neq0,
\end{equation}
leaving a nonvanishing boundary term and thus a nonvanishing Poisson bracket. The background's presence thus obstructs momentum conservation, which further justifies our choice of $c=0$.

Together these constraints have a clear physical meaning. Momentum conservation demands that translations of our surface $\partial V$ cannot change the charge, which is enforced by the boundary condition $\hat{n}_i \partial_k \rho_i|_{\partial V}=0$. Translational invariance further prevents our charge $\rho_i$ from having local sources or sinks, i.e. $\partial_i\rho_i=0$ (there cannot be magnetic monopoles). 

We repeat the naive argument in the previous section, except we now impose our Gauss's law constraint via a Lagrange multiplier $\lambda$. The constitutive relations are thus
\begin{equation}
    \partial_i \mu_j J_{ji} +\partial_i v_j \tau_{ji} = \partial_i s_i +\lambda \partial_i \rho_i =\frac{\partial s_i}{\partial v_j}\partial_i v_j +\frac{\partial s_i}{\partial \mu_j}\partial_i \mu_j + \lambda\phi_{ij}\partial_i v_j +\lambda \chi_{ij} \partial_i \mu_j,
\end{equation}
where the last step follows from the chain rule. We have defined the generalized thermodynamic susceptibilities $\phi_{ij}\coloneq \partial \rho_i/\partial v_j$ and $\chi_{ij}\coloneq \partial \rho_i/\partial \mu_j$. Note that $\chi_{ij}=\chi_{ji}$ since $\rho_i \equiv \partial P/\partial \mu_i$ ensures that $\chi_{ij}$ is a symmetric tensor composed of the second derivatives of the pressure. Equating the motifs provides the constitutive relations
\begin{equation}
  \tau_{ji}= \frac{\partial s_i}{\partial v_j } +\lambda \phi_{ij}, \quad   J_{ji} = \frac{\partial s_i}{\partial \mu_j} + \lambda \chi_{ij}.
\end{equation}
 We shall see that incorporating the Lagrange multiplier $\lambda$ removes the contradiction in the previous section. Under a similar ansatz $s_i =Pv_i -A_{ki}\rho_k$, we first evaluate our current:
\begin{equation}
     J_{ji}= \rho_j v_i- \frac{\partial A_{ki}}{\partial \mu_j}\rho_k -A_{ki}\chi_{kj}  +\lambda \chi_{ij}.
\end{equation}
Since $J_{ji}$ is antisymmetric, contracting with an $\SON{3}$-invariant tensor $T_{ij,mn}$ that is symmetric in $(ij)$ and $(mn)$ and also between swapping $(ij)\leftrightarrow (mn)$ will cause both sides to vanish:
\begin{equation}
    0 =T_{ij,mn}\left(\rho_j v_i- \frac{\partial A_{ki}}{\partial \mu_j}\rho_k -A_{ki}\chi_{kj}  +\lambda \chi_{ij}\right).
\end{equation}
Equating terms with and without a susceptibility $\chi_{ij}$ yields two conditions:
\begin{equation}
    A_{ki}=\lambda\delta_{ki}, \quad  T_{ij,mn}\rho_j v_i =T_{ij,mn}\frac{\partial A_{ki}}{\partial \mu_j}\rho_k.
\end{equation}
Plugging $A_{ki}$ into the above produces
\begin{equation}
    T_{ij,mn}\rho_j v_i =T_{ij,mn}\frac{\partial \lambda }{\partial \mu_j}\rho_i.
\end{equation}
Since every symmetric $\SON{3}$-invariant tensor can be decomposed as 
\begin{equation}
     T_{ij,mn} =c_1\left(\delta_{im}\delta_{jn}+\delta_{in}\delta_{jm}-\frac{2}{3}\delta_{ij}\delta_{mn}\right) +c_2 \delta_{ij}\delta_{mn},
\end{equation}
we can plug that into the above and observe that $\lambda $ must be $\lambda = \mu_iv_i$. This yields (after doing a similar calculation for $\tau_{ji}$)
\begin{equation}
    \tau_{ji} =P\delta_{ij}+g_jv_i-\mu_j\rho_i, \quad J_{ji}=\rho_j v_i -\rho_iv_j.
\end{equation}
We have successfully constructed the non-dissipative constitutive relations for non-relativistic magnetohydrodynamics; we are not aware of earlier work deriving this effective field theory for explicitly non-relativistic systems.  Note that our relations closely mirror those of the relativistic case (Ref.~\cite{Grozdanov_2017}), with the primary differences arising from our framework's lack of boost invariance.

\section{Consistency Conditions and the Generalized Gauss's Law for Generic Momentum Conserving Theories}
\label{Sec:consistencyconditions}
In this section, we derive the spatial derivative constraint on $\rho_a$ imposed by momentum conservation and subsequently utilize it to solve the ideal hydrodynamics. For the hydrodynamic theory to be consistent, we further assert two representation-theoretic consistency conditions that must be satisfied.

\subsection{The First Consistency Condition (C1) }
To proceed, we shall assume that if the differential $\Gamma_{ia\alpha}\partial_i f_a=0$ has an infinite family of non-trivial solutions, they may be written in the form $f_a= \Gamma_{j_1\cdots j_m aI}\partial_{j_1}\cdots \partial_{j_m} g_I $, where $g_I$ transforms under the irreducible representation $R_I$ (Assumption~\ref{as:exact}). Such an assumption holds in all theories we are aware of (such as higher-form symmetries). Note that there may be multiple irreducible representations $R_I$ that satisfy the equation $\Gamma_{ia\alpha}\partial_i f_a = 0$; for now, we only concern ourselves with the existence of at least one such valid $R_I$. Such multi-derivative solutions exist in theories like fracton magnetohydrodynamics \cite{fractonMHD}, where $R_a=\textbf{5}$ (the symmetric traceless representation) and $f_{ij}=(\partial_i\partial_j -\frac{1}{3}\delta_{ij}\partial^2) \phi $, for a scalar field $\phi$. As we will demonstrate in Eq.~\eqref{eq:momentumconstraint}, this exactness assumption is quite powerful because the same differential operator $\Gamma_{j_1 \cdots j_m a I} \partial_{j_1} \cdots \partial_{j_m}$ will control the spatial derivative constraint on the charge density $\rho_a$.

Given that only a single derivative constraint can be utilized to restrict the current (see Section \ref{Sec:obstructionansatz}) and $m$ will label the order of the derivative constraints on $\rho_a$, we shall restrict ourselves for the remainder of this paper to theories in which $R_a$, $R_\alpha$, and $R_I$ are chosen such that $m=1$. We postulate that theories with $m\geq 2$ will generically exhibit the spontaneous symmetry breaking of emergent conservation laws; we leave a detailed confirmation to future work (also see the derivative-counting nonlocality argument at the beginning of Section~\ref{Sec:secondconsistencycondition}). Substituting our solution for $f_a$ into the prior vanishing condition gives
\begin{equation}
\Gamma_{ia\alpha}\Gamma_{jaI}\partial_i\partial_{j} g_I=0.
\end{equation}
Because partial derivatives commute, this condition requires that the completely symmetric part of the product $ \Gamma_{ia\alpha}\Gamma_{j aI}$ vanishes over the indices $(i,j)$. More formally, if we view  $\Gamma_{ia\alpha}\Gamma_{j aI}$ as a map on $\mathcal{V}^{\otimes 2}$, it must vanish on the symmetric subspace $\text{Sym}^{2}(\mathcal{V})\subseteq \mathcal{V}^{\otimes 2}$, which is spanned by vectors of the form $v_1\otimes v_2+v_2\otimes v_1$.  The product $\Gamma_{ia\alpha}\Gamma_{jaI}$ acts as an intertwiner mapping  $\mathcal{V}^{\otimes 2}\to (\oplus_k R_{\alpha_k}) \otimes R_I$. By Schur's lemma, an intertwiner can only map to isomorphic irreducible representations. Therefore, the image of $\Gamma_{ia\alpha}\Gamma_{j aI}(\text{Sym}^{2}(\mathcal{V}))$ must consist exclusively of irreducible representations present in the irrep decomposition of $\text{Sym}^{2}(\mathcal{V})$. To guarantee that this mapping has no symmetric part, we strictly require that $\Gamma_{ia\alpha}\Gamma_{j aI}(\text{Sym}^{2}(\mathcal{V}))=0$.  This is generally satisfied if the irrep decomposition of $(\oplus_kR_{\alpha_k})\otimes R_I$ contains no completely symmetric components, meaning
\begin{equation}
 \label{eq: firstcondition}
  \text{Sym}^{2}(\mathcal{V})\cap \left[(\oplus_kR_{\alpha_k})\otimes R_I \right]=\emptyset. 
\end{equation} 
We shall call the above condition our \textbf{First Consistency Condition (C1)} for formulating theories compatible with momentum conservation that avoid spontaneous symmetry breaking. In Section \ref{sec:RIchoice}, we will explicitly analyze which choices of $R_I$ satisfy this condition. We note that this condition becomes increasingly restrictive as more irreducible representations are present in $R_\alpha$, since all irreducible representations must simultaneously satisfy our constraint. 

\subsection{Momentum Conservation Imposes a Generalized Gauss's Law Constraint}
\label{sec:generalizedgausslawsec}
We now ask: which theories are compatible with momentum conservation without acquiring a Goldstone mode? Imposing momentum conservation will require that the Poisson bracket $\{Q[f_a],P_k\}=0$ for any $k$. If this bracket were non-trivial, translation would generate a new, non-trivial conserved quantity, signaling a potential Goldstone mode (as occurs, for example, in the theory with dipole and momentum conservation \cite{Glorioso_2022, Glorioso_2023}). We observed in Section \ref{Sec:obstructionansatz} that a constraint on the density $\rho_a$ is \textit{required} for ideal hydrodynamics to be consistent. In the one-form case, such a constraint was found to be a Gauss's law constraint $\partial_i\rho_i=0$ (see Section \ref{Sec:IdealMHD}). In this section, we show that for a general density $\rho_a$, the density constraint becomes a \textit{generalized} Gauss's law. We incorporate this generalized Gauss's law into the MSR framework in Section \ref{Sec:secondconsistencycondition} to derive the Second Consistency Condition (C2) for momentum conservation. 

To ensure invariance under translations and prevent spontaneous symmetry breaking, we demand
\begin{equation}
\label{eq:momentumconstraint}
    0= \{P_k, Q[f_a;V]\} = \int_V d^dx  \Gamma_{j aI}\partial_{j} g_I \partial_k \rho_a = \int_{\partial V} dS \hat{n}_j \Gamma_{jaI}g_I \partial_k\rho_a -\int_V d^dx \Gamma_{j aI}g_I \partial_{j} \partial_k \rho_a,
\end{equation}
where we have integrated by parts to shift the derivatives off $g_I$. Our argument is now identical to the one-form case. Selecting a family of $g_I$ that vanish on the boundary forces the bulk integrand to vanish and thus $\Gamma_{jaI} \partial_j \partial_k \rho_a = 0$. Since $k$ is an arbitrary spatial index, $\Gamma_{jaI} \partial_j \rho_a$ has a vanishing gradient and thus is constant in space: $\Gamma_{jaI} \partial_j \rho_a = c_I(t)$, where $c_I(t)$ still in principle depends on time. We can apply $\Gamma_{jaI}\partial_j$ to the continuity equation to show that $c_I(t)$ is also constant in time:
\begin{equation}
    0 = \partial_t(\Gamma_{jaI}\partial_j\rho_a) + \Gamma_{jaI}\Gamma_{ia\alpha}\partial_j \partial _iJ_\alpha=\partial_t(\Gamma_{jaI}\partial_j\rho_a) =\partial_t(c_I),
\end{equation}
where $\Gamma_{jaI}\Gamma_{ia\alpha}\partial_j \partial _iJ_\alpha=0$ by C1. We can thus take $c_I=0$ without loss of generality since it is dynamically frozen (by  Assumption~\ref{as:boundary}), generating the generalized Gauss's law constraint
\begin{equation}
\Gamma_{ja I}\partial_{j}\rho_a=0.
\end{equation}
After we have chosen our Gauss's law constraint, our bracket becomes just a boundary term for generic $g_I$:
\begin{equation}
     \{P_k, Q[f_a;V]\}= \int_{\partial V} dS \hat{n}_j \Gamma_{jaI}g_I \partial_k\rho_a,
\end{equation}
which itself vanishes for \textit{arbitrary} $g_I$. This holds if we either restrict ourselves to functions $g_I$ satisying $g_I|_{\partial V}=0$ or if we impose the boundary condition $\hat{n}_j \Gamma_{jaI} \partial_k\rho_a|_{\partial V}=0$. We again adopt the latter (Assumption~\ref{as:boundary}) to allow non-trivial charge configurations (the derivation is analogous to the one-form case and is thus omitted). 

We further note that this generalized Gauss's law constraint restricts our conserved charges to lie on the boundary, since 
\begin{equation}
    Q[f_a] =  \int_V d^dx \Gamma_{jaI}\partial_j g_I \rho_a = \int_{\partial V}dS\hat{n}_j\Gamma_{jaI} g_I \rho_a - \int_V d^dx \Gamma_{jaI} g_I \partial_j\rho_a = \int_{\partial V}dS\hat{n}_j\Gamma_{jaI} g_I \rho_a,
\end{equation}
which leaves only the surface term. Momentum conservation thus imposes that all our charge is restricted to the boundary $\partial V$. Physically, states with a nonzero global charge $Q\neq 0$ can only be accommodated by formulating the theory on a finite-volume manifold with a boundary. This restriction allows us to differentiate whether a theory has spontaneous symmetry breaking forcing it into the zero charge sector. Such a case may occur when considering dipolar charges, in which $\{P_k,Q_{\mathrm{dipole}}\}\propto Q_{\mathrm{charge}}$.

\subsection{The Second Consistency Condition (C2)}
\label{Sec:secondconsistencycondition}

We now encode the above constraints into the MSR effective field theory framework (Assumption~\ref{as:msr}), following a procedure analogous to Section~\ref{Sec:obstructionansatz}. However, we now explicitly enforce the constraint $\Gamma_{iaI}\partial_i\rho_a=0$ via a Lagrange multiplier field $\lambda_I$. Inserting our choice for $H$ into Eq.~\eqref{eq:TRS} and adding our Lagrange multiplier, we arrive at
\begin{equation}
     (\partial_i \mu_a ) J_{ai} + (\partial_i v_j )\tau_{ji} =\partial_i s_i +\lambda_I \Gamma_{iaI}\partial_i\rho_a = \partial_i s_i +\lambda_I \Gamma_{iaI}\phi_{aj} \partial_iv_j + \lambda_I \Gamma_{iaI}\chi_{ab} \partial_i \mu_b, 
\end{equation}
where the last step follows from the chain rule. We have defined the generalized thermodynamic susceptibilities $\phi_{aj}\coloneq \partial \rho_a/\partial v_j$ and $\chi_{ab}\coloneq \partial \rho_a/\partial \mu_b$. Note that $\chi_{ab}=\chi_{ba}$ since $\rho_a \equiv \partial P/\partial \mu_a$ ensures that $\chi_{ab}$ is a symmetric tensor composed of the second derivatives of the pressure. This form immediately justifies why we set $m=1$: because $\partial_is_i$ contains only a single spatial derivative, $m\geq 2$ would yield higher-order spatial derivatives that cannot consistently modify the constitutive relations at ideal order. These theories may generically require $J_{ai}$ and/or $\tau_{ji}$ to be integrals of our hydrodynamic densities at ideal order, violating Assumption~\ref{as:locality}. Just as we did in our naive analysis, we equate the coefficients of the independent gradients $\partial_i v_j$ and $\partial_i \mu_a$. This yields the modified constitutive relations
\begin{equation}
    \tau_{ji} = \frac{\partial s_i}{\partial v_j} +\lambda_I \Gamma_{iaI} \phi_{aj}, \quad J_{ai} = \frac{\partial s_i}{\partial \mu_a} + \lambda_I \Gamma_{ibI}\chi_{ab}.
\end{equation}
The addition of the Lagrange multiplier term modifies the constitutive relations from Section \ref{Sec:obstructionansatz}, coupling the current and stress tensor to the relevant thermodynamic susceptibilities. To determine the Lagrange multiplier $\lambda_I$ alongside the explicit form of the stress tensor and current, we utilize the same hydrodynamic ansatz $s_i =Pv_i- A_{ai}\rho_a$ (Assumption \ref{as:entropy}). Substituting this ansatz into our expression for $J_{ai}$ gives
\begin{equation}
    J_{ai} = \rho_a v_i  -\frac{\partial A_{bi}}{\partial \mu_a} \rho_b -A_{bi} \chi_{ab} +\lambda_I \Gamma_{ibI} \chi_{ab}. 
\end{equation}
Projecting this current into the complementary representation $R_{\tilde{\alpha}}$ by applying the Clebsch-Gordan coefficients $\Gamma_{ia\tilde{\alpha}}$ completely annihilates $J_{ai}$:
\begin{equation}
    \Gamma_{ia \tilde{\alpha}}\left( \rho_a v_i  -\frac{\partial A_{bi}}{\partial \mu_a} \rho_b -A_{bi} \chi_{ab} +\lambda_I \Gamma_{ibI} \chi_{ab}\right) =0.
\end{equation}
Because this condition must hold for arbitrary densities and potentials, we can independently equate the respective $\rho$ terms and susceptibility ($\chi$) terms. This produces two simultaneous conditions:
\begin{equation}
    \Gamma_{ia\tilde{\alpha}} \rho_a v_i = \Gamma_{ia\tilde{\alpha}}  \frac{\partial A_{bi}}{\partial \mu_a} \rho_b, \quad \Gamma_{ia\tilde{\alpha}}\Gamma_{ibI}\lambda_I\chi_{ab} =\Gamma_{ia\tilde{\alpha}}A_{bi}\chi_{ab}.
\end{equation}
These conditions are solved if we set $ A_{bi} = \lambda_I \Gamma_{ib I}$ and set the Lagrange multiplier $\lambda_{I} = \Gamma_{jaI}\mu_a v_j$. This specific choice of $\lambda_I$ follows from the decomposition
\begin{equation}
    \Gamma_{jaI}\Gamma_{ibI}=\delta_{ij} \delta_{ab}-K^{(I)}_{ia,jb},
\end{equation}
where $K^{(I)}_{ia,jb}$ is a complicated combination of projectors onto all irreducible representations in the decomposition for $\mathcal{V}\otimes R_{a}$. Namely, we can always rewrite
\begin{equation}
    \Gamma_{jaI}\Gamma_{ibI}= \sum_{R\subseteq \mathcal{V}\otimes R_a} \kappa_I(R) P_{ia,jb}^{(R)},  
\end{equation}
where $\kappa_I(R)$ is some real coefficient, $P_{ia,jb}^{(R)}=(\text{dim}(R)/\lVert \Gamma_{iaR} \rVert^2)\Gamma_{iaR}\Gamma_{jbR}$ is the projector onto the irreducible representation $R$, and $\lVert \Gamma_{iaR} \rVert=\sqrt{\Gamma_{iaR}\Gamma_{iaR}}$. Since $\delta_{ij}\delta_{ab}=\sum_{R\subseteq \mathcal{V}\otimes R_a} P_{ia,jb}^{(R)}$, we can take $K^{(I)}_{ia,jb}$ to be
\begin{equation}
    K^{(I)}_{ia,jb} = \sum_{R\subseteq \mathcal{V}\otimes R_a}(1- \kappa_I(R)) P_{ia,jb}^{(R)}.
\end{equation}
Furthermore, for this constraint to be self-consistent, it must be the case that projecting $K$ into the complementary representation $R_{\tilde{\alpha}}$ vanishes, i.e. $\Gamma_{ia\tilde{\alpha}} K^{(I)}_{ia,jb}=0$. Applying the above decomposition, this is equivalent to requiring 
\begin{equation}
   \Gamma_{jaI}\Gamma_{ibI} \Gamma_{ia\tilde{\alpha}} = \Gamma_{jb\tilde{\alpha}}.
\end{equation}
By evaluating the contraction of the generalized Clebsch-Gordan coefficients, this latter condition can be written in terms of Wigner $6j$ symbols for the group $G=\SON{d}$ \cite{birdtracks}:
\begin{equation}
    \Gamma_{jaI}\Gamma_{ibI} \Gamma_{ia\tilde{\alpha}} = (-1)^{\Phi(\tilde{\alpha}, I,\mathcal{V})} \lVert \Gamma_{jaI} \rVert^2\begin{Bmatrix} j_{\tilde{\alpha}} & j_{\mathcal{V}} & j_a \\ j_{I} & j_{\mathcal{V}} & j_a \end{Bmatrix}_{\SON{d}}\Gamma_{jb \tilde{\alpha}}\equiv   \kappa_I(R_{\tilde{\alpha}})\Gamma_{jb \tilde{\alpha}},
\end{equation}
where $j_a$ denotes the highest weight labels for the irreducible representation $R_a$, $\Phi(\tilde{\alpha}, I,\mathcal{V})$ is a phase factor dependent on the involved representations and $\{\cdot \}_{\SON{d}}$ denotes the Wigner $6j$ symbol for $\SON{d}$. This entire phase-weighted $6j$ term is exactly the recoupling coefficient $\kappa_I(R_{\tilde{\alpha}})$. To see this, we can apply the projective form of $\Gamma_{jaI}\Gamma_{ibI}$ to $\Gamma_{ia\tilde{\alpha}}$, which yields $\kappa_I(R_{\tilde{\alpha}})\Gamma_{jb\tilde{\alpha}}$. The factor of $ \lVert \Gamma_{jaI} \rVert^2$ arises from us normalizing each Clebsch-Gordan coefficient by dividing by its magnitude and observing that the $R_{\tilde{\alpha}}$ normalizations will cancel. For the specific case when $G=\SON{3}$, the explicit coefficient is 
\begin{equation}
  \kappa_I(R_{\tilde{\alpha}})= \lVert \Gamma_{jaI} \rVert^2\begin{Bmatrix} j_{\tilde{\alpha}} & j_{\mathcal{V}} & j_a \\ j_{I} & j_{\mathcal{V}} & j_a\end{Bmatrix}_{\SON{3}}.
\end{equation}
 The phase above is trivial because we need to perform an even number of index swappings to obtain our Clebsch-Gordan coefficients in the proper order. Our consistency condition is thus equivalent to showing that $  \kappa_I(R_{\tilde{\alpha}})$ is identical for all irreducible representations in $R_{\tilde{\alpha}}$ (while strict equality requires $  \kappa_I(R_{\tilde{\alpha}})=1$, our normalization convention can be adjusted so long as $\kappa$ is identical for all $R_{\tilde{\alpha}}$). If this condition cannot be satisfied for any valid combination of $R_I$, $R_{\tilde{\alpha}}$, and $R_a$, it signals an obstruction to hydrodynamics. We designate this as our \textbf{Second Consistency Condition (C2)} for valid theories that conserve momentum without an emergent Goldstone mode. We note that, unlike C1, C2 becomes \emph{less} restrictive as there are more irreducible representations in $R_\alpha$, as there are fewer possible $R_{\tilde{\alpha}}$ values and thus fewer constraints. In this case, we can view $K^{(I)}_{ia,jb}$ as a weighted projector onto $R_\alpha$:
\begin{equation}
     K^{(I)}_{ia,jb} = (1- \kappa_I(R_\alpha)) P_{ia,jb}^{(\alpha)},
\end{equation}
where there is an implied sum over $R_\alpha$ if it is reducible. In Appendix \ref{Appdx:enumeratingtheoriesSO3}, we demonstrate that for the specific case of $\SON{3}$, the only valid theories satisfying both consistency conditions are $p$-form symmetries. While $p$-form symmetries naturally satisfy these conditions for generic $\SON{d}$ (as detailed in Section \ref{Sec:pformframewwork}), we will demonstrate in Section \ref{Sec:nothingelse} that there do not exist theories beyond $p$-form symmetries that satisfy our two conditions, given sufficiently large $d$. 

If both consistency conditions are satisfied, we obtain explicit solutions for $J_{ai}$ and $\tau_{ji}$ by plugging these choices for $A_{bi}$ and $\lambda_I$ into our constitutive relations:
\begin{equation}
    \tau_{ji} = P\delta_{ij} +g_j v_i -\rho_a \mu_a \delta_{ij} +K^{(I)}_{ia,jb} \rho_b\mu_a,  \quad  J_{ai} = K^{(I)}_{ia,jb}\rho_b v_j.
\end{equation}

\section{On the Possible Choices for \texorpdfstring{$R_I$}{R\_I}}
\label{Sec:PossibleRI}

In Section \ref{Sec:consistencyconditions}, we mentioned the existence of the representation $R_I$ that solves the equation $\Gamma_{ia\alpha}\partial_i f_a=0$, but provided no further detail. It is, however, possible to entirely determine the possible choices of $R_I$ from the representations $R_a$ and $(\oplus_k R_{\alpha_k})$. To do this, we will represent each irreducible representation of $\SON{d}$ by its corresponding \emph{Young tableau} and then utilize the \emph{Newell-Littlewood} rules to write down the irrep decomposition of the product of two irreducible representations, such as $R_a \otimes R_\alpha$. We shall first provide a review of these rules before deriving the possible choices for $R_I$. Importantly, because irreducible representations are typically uniquely determined by their Young tableaux in $\SON{d}$, we shall use the representation and its tableau label interchangeably for the rest of the paper (modulo a minor subtlety for even $d$ that is addressed below.)
\subsection{A Review of Young Tableaux and the Newell-Littlewood Rules for \texorpdfstring{$\SON{d}$}{SO(d)}}
\label{Sec:reviewyoung}

It is often convenient to have a graphical representation for various irreducible representations. A \textit{Young tableau} (or \textit{Young diagram}) provides such a graphical framework \cite{Georgi1999YoungTableaux}. Each Young tableau comprises $L$ rows of left-aligned boxes in non-increasing order. We label the shape of each Young tableau by $[\lambda]_L=[\lambda_1,\lambda_2\cdots, \lambda_L]$, where $\lambda_i$ is the number of boxes in the $i$-th row. For example, the Young diagram $[2,2,1]$ will label the diagram
\begin{equation}
   [2,2,1]= \vcenter{\hbox{\ydiagram{2,2,1}}}.
\end{equation}
Often it is convenient to denote $m$ rows of the same length $\lambda$ by $[\lambda^m] $; going back to the previous example, we could write $[2,2,1] = [2^2,1]$. In $\SON{d}$, each box of a Young tableau will correspond to a $d$-dimensional vector ($[1]=\mathcal{V}$), so a rank-$n$ tensor will have $n$ boxes in its Young tableau. Suppose we wanted to build an irreducible tensor $T_{ijk}$ by first symmetrizing the indices $i$ and $j$ and then antisymmetrizing the indices $i$ and $k$. Indices placed in the same row of a Young tableau are completely symmetrized, whereas indices in the same column are completely antisymmetrized. The corresponding Young tableau for $T_{ijk}$ shall be $[2,1]$, which encodes a pair of symmetric indices and a pair of antisymmetric indices overlapping at a single shared index. If we were working in $\SUN{d}$, we would be done. For $\SON{d}$, however, the resulting Young tableau will not necessarily represent an irreducible representation. This is because the metric $\delta_{ij}$ is an invariant tensor of $\SON{d}$. We can thus contract $\delta_{ij}T_{ijk}$ as an operation in $\SON{d}$ and create a smaller tensor without breaking rotational invariance. For example, consider the symmetric tensor $T_{ij}$. In order for $T_{ij}$ to be irreducible, we must break it up into a traceless component and its trace: $T_{ij}=(T_{ij}-\frac{1}{d}\delta_{ij} T_{kk})+\frac{1}{d}\delta_{ij} T_{kk} $. In order for the Young tableau $[2]$ to truly represent an irreducible representation of $\SON{d}$, we must assign it to the symmetric \textit{traceless} component $(T_{ij}-\frac{1}{d}\delta_{ij} T_{kk})$, rather than the purely symmetric component.  

We now explain how to perform the irrep decomposition for two Young tableaux $[\lambda]_L\otimes [\sigma]_M$ in $\SON{d}$. First, we remind the reader of the rules for $\SUN{d}$; the rules for $\SON{d}$ will follow \cite{Georgi1999YoungTableaux}. In $ [\sigma]_M$, write the letter $a$ in each of the boxes in the first row, the letter $b$ in each of the boxes in the second row, etc. until the Young tableau $[\sigma]_M $ is full. Then take each of the $a$ boxes in $ [\sigma]_M$ and place them in $[\lambda]_L$ such that the new tableau is a valid Young tableau (rows are nonincreasing in length as one descends the tableau) and no two $a$s appear in the same column. Likewise, do the same with the $b$ rows, $c$ rows, etc. until the entire $[\sigma]_M$ tableau is exhausted. Then, reading the boxes from right to left and top to bottom, the number of $a$s must be greater than or equal to the number of $b$s which must be greater than or equal to the number of $c$s, and so on at every step. Any tableau that violates this should be erased. The resulting sum of valid tableaux provides the irrep decomposition in $\SUN{d}$.  Finally, if a column has a length greater than $d$, the diagram must identically vanish because an antisymmetric tensor of rank greater than $d$ must have a repeated index. These are called the \textit{Littlewood-Richardson rules} for $\SUN{d}$. For example, consider $[1^2]\otimes [1^2]$ for $d>4$. The decomposition is given by

\begin{equation}
    \vcenter{\hbox{\ydiagram{1,1}}} \otimes \vcenter{\hbox{\begin{ytableau}
a \\
b
\end{ytableau}}} \to \vcenter{\hbox{
  \begin{ytableau}
  {} \\
  {} \\
  a \\
  b
  \end{ytableau}}
}
\oplus
\cancel{\vcenter{\hbox{
  \begin{ytableau}
  {} & b \\
  {} \\
  a
  \end{ytableau}
}}}
\oplus
\cancel{\vcenter{\hbox{
  \begin{ytableau}
  {} & a & b \\
  {}
  \end{ytableau}
}}}
\oplus
\vcenter{\hbox{\begin{ytableau}
{} & a \\
{} & b
\end{ytableau}}}
\oplus
\vcenter{\hbox{\begin{ytableau}
{} & a \\
{} \\
b
\end{ytableau}}},
\end{equation}
where we have written slashes through each diagram that is not allowed. For $\SON{d}$, we must also account for traces over $\delta_{ij}$. The resulting rules are called \textit{Newell-Littlewood} rules. First, we perform the exact same decomposition rules for the full $[\lambda]_L\otimes [\sigma]_M$ in $\SUN{d}$. Then we take the trace by removing a single box from both $[\lambda]_L$ and $[\sigma]_M$ such that the two new sub-tableaux remain valid and then perform the $\SUN{d}$ Littlewood-Richardson rules on these new tableaux. After enumerating all single-trace tableaux (removing a single box from each), we remove two boxes from $[\lambda]_L$ and $[\sigma]_M$. Because contracting a symmetric pair of indices with an antisymmetric pair vanishes, the symmetry of the removed boxes must match: if the two boxes removed from $[\lambda]_L$ sit in the same row/column, the corresponding two boxes removed from $[\sigma]_M$ must also sit in the same row/column. We again perform the Littlewood-Richardson decomposition on these doubly-traced tableaux. We continue this process until either $[\lambda]_L$ or $[\sigma]_M$ is fully exhausted and sum over all possible resultant tableaux. 

 We must now apply the modification rules derived in Ref.~\cite{Kingrules}. First, we define the conjugate tableau $[\lambda']_L$ to be the transpose of $[\lambda]_L$. The $i$-th column length will be given by $\lambda_i'$. For orthogonal groups, a Young tableau $[\lambda]_L$ with $\lambda_1'\leq \lfloor d/2\rfloor$ is called \textit{standard}. If a tableau is nonstandard, it can be made standard via an equivalent dual tableau by contracting the first column with the Levi-Civita tensor $\epsilon^{j_1\cdots j_d}$ (known as \textit{Hodge duality}, which maps an antisymmetric $p$-index tensor into an antisymmetric $(d-p)$-index tensor) and then contracting with the metric tensor $\delta_{jk}$. This converts $[\lambda]_L$ into a tableau whose first column has length $d-L$.  Graphically, this is enforced by removing a hook of length $2L-d$ from the original Young tableau: start at the bottom box of the offending column, trace a path of length $2L-d$ up and right along the boundary of the Young tableau, and remove it. If the resulting Young tableau is no longer valid, we say the tableau vanishes and does not contribute. We multiply the resulting tableau by a phase factor $(-1)^x$, where the \emph{depth} $x$ is the difference in the starting and ending column indices of the removed hook. We henceforth dub this hook modification rule the ``snake rule". There is one final modification rule if $d=2k$. If we apply the aforementioned snake rule and $L=k$, applying the Levi-Civita symbol $\epsilon^{j_1\cdots j_d}$ maps the Young tableau back into itself. Applying it twice leaves the Young tableau invariant, meaning this action multiplies the Young tableau by either $+1$ or $-1$. This tableau is thus reducible into two separate irreducible representations: a self-dual part (corresponding to $ +1$) or an anti-self-dual part (corresponding to $-1$). The self-dual part is denoted by $[\lambda]_L^{+}= [\lambda_1,\cdots, \lambda_L]$ while the anti-self-dual part is denoted by  $[\lambda]_L^{-}= [\lambda_1,\cdots, -\lambda_L]$. This dichotomy reflects that $\SON{2k}$ is classified as a Lie group of type $D_k$, whereas $\SON{2k+1}$ is of type $B_k$ \cite{Koike1987YoungdiagrammaticMF}. Throughout this paper, we always label representations by their standard tableaux and the generalized indices in $\rho_a$, $J_\alpha$, etc.\ correspond to the properly dualized standard tableaux. For instance, consider the two-form theory in $\SON{3}$ (as seen in Appendix~\ref{Appdx:enumeratingtheoriesSO3}). The two-form density $\rho_a$ is labeled by the tableau $[1]$ and it carries a Hodge-dual vector index. If $R_I$ is the ordinary undualized vector irreducible representation (also labeled by the tableau $[1]$), the generalized Clebsch-Gordan coefficient is the Levi-Civita tensor $\Gamma_{jaI}=\epsilon_{jaI}$.
 
Before moving on, let us demonstrate an example. Let us again look at $[1^2]\otimes [1^2]$, but this time in $\SON{5}$. Before applying the modification rules, we will again have the same three $\SUN{d}$ diagrams as before, but we now include the traced components:

\begin{equation}
   [1^2]\otimes [1^2]=  \vcenter{\hbox{\ydiagram{1,1}}} \otimes \vcenter{\hbox{\begin{ytableau}
a \\
b
\end{ytableau}}} \to  \vcenter{\hbox{
  \begin{ytableau}
  {} \\
  \tikz[remember picture, baseline=(a1.base)] \node[inner sep=0] (a1) {}; \\
  a\\
\tikz[remember picture, baseline=(b1.base)] \node[inner sep=0] (b1) {$b$};
  \end{ytableau}}
} \oplus
\vcenter{\hbox{\begin{ytableau}
{} & a \\
{} & b
\end{ytableau}}}
\oplus
\vcenter{\hbox{\begin{ytableau}
{} & a \\
{} \\
\tikz[remember picture, baseline=(b2.base)] \node[inner sep=0] (b2) {$b$};
\end{ytableau}}} \oplus \left(\vcenter{\hbox{\begin{ytableau}
{} \\
\end{ytableau}}}\otimes \vcenter{\hbox{\begin{ytableau}
a 
\end{ytableau}}}\right),
\begin{tikzpicture}[remember picture, overlay]
    \draw[red, thick] ($(a1.north)+(0,0.4ex)$) -- ($(b1.south)-(0,0.4ex)$);
    
    \draw[red, thick] ($(b2.north)+(0,0.4ex)$) -- ($(b2.south)-(0,0.4ex)$);
\end{tikzpicture}
\end{equation}
where we have drawn lines through boxes that will be removed by the snake rule. If we further evaluate the trace decomposition rules to $[1]\otimes [1]$ and execute the snake rule on the nonstandard tableaux, we will arrive at the decomposition
\begin{equation}
    [1^2]\otimes [1^2]= \vcenter{\hbox{\ydiagram{1}}} \oplus \vcenter{\hbox{\ydiagram{2,2}}} \oplus \vcenter{\hbox{\ydiagram{2,1}}} \oplus \vcenter{\hbox{\ydiagram{1,1}}} \oplus \vcenter{\hbox{\ydiagram{2}}} \oplus [0],  
\end{equation}
where we dropped the labels $a$ and $b$ and $[0]$ denotes an empty box (the scalar irreducible representation), arising from the full contraction of both indices. While $\SON{3}$ has a unique decomposition in terms of irreducible representations, this breaks down for generic $d>3$, where representations can appear with higher multiplicities. This makes computing our Wigner $6j$ symbols much harder, as we generally need to keep track of exactly which multiplicity copy of an irreducible representation we are considering.

\subsection{Determining the Allowed Representations \texorpdfstring{$R_I$}{R\_I}}
\label{sec:RIchoice}
 Our goal is to determine, for $\SON{d}$, which representations $R_I$ can be chosen such that $R_\alpha \otimes R_I$ contains no purely symmetric rank-2 component (i.e., the $[2]$ or $[0]$ tableau). First, note that the condition $\Gamma_{ia\alpha }\partial_i f_a=0$ requires that $R_a \subseteq [1]\otimes R_\alpha$ (where $\subseteq$ denotes inclusion in the irrep decomposition and $[1]$ is the vector representation). Likewise, the condition $f_a = \Gamma_{jaI} \partial_j g_I$ requires that $R_I \subseteq [1]\otimes R_a   \subseteq   [1]\otimes [1]\otimes R_\alpha $. Nontrivial choices for $R_a$, $R_\alpha$, and $R_I$ must respect these constraints.

Our problem can thus be rephrased as: in the irrep decomposition of $R_\alpha\otimes R_I$, when does $[2]$ or $[0]$ have multiplicity zero? At this stage of the analysis, (before doing modification rules), the only way $[0]$ can occur in the decomposition for $R_\alpha\otimes R_I$ is if $R_I=R_\alpha$ (we must be able to subtract off the entire representation doing traces, and thus the representations must be equal). We thus choose $R_I\neq R_\alpha$ and ignore the $[0]$ case for the rest of the paper (indeed, this will only limit the number of possible theories further). The trick for the $[2]$ case is to note that the multiplicity of $R_I$ in $R_\alpha^*\otimes [2]$ is identical to the multiplicity of $[2]$ in $R_\alpha\otimes R_I$. Specifically, for an irreducible representation $A$ and a generic representation $B$, the dimension $\text{dim}(\text{Hom}_G(A,B))$ provides the multiplicity $m$ of $A $ in $B$. From Fulton and Harris \cite{FultonHarris1991} Eq.~2.9, this dimension is given by the inner product of their characters:
\begin{equation}
    m = \langle \chi_A,\chi_B\rangle,
\end{equation}
where $\chi_A$ is the \textit{character} of $A$. In our case, $A=[2]$ and $B=R_\alpha\otimes R_I$. Because characters are multiplicative ($\chi_{R_\alpha\otimes R_I}= \chi_{R_\alpha}\chi_{R_I}$), we evaluate
\begin{equation}
    m= \langle \chi_{[2]},\chi_{R_\alpha}\chi_{R_I}\rangle= \langle \overline{\chi}_{R_\alpha}\chi_{[2]},\chi_{R_I}\rangle =\langle \chi_{R_\alpha^{*}}\chi_{[2]},\chi_{R_I}\rangle = \langle \chi_{R_I},\chi_{R_\alpha^*\otimes [2]}\rangle,
\end{equation}
where we used the fact that $\overline{\chi}_{R_\alpha}= \chi_{R_\alpha^*} $ and that multiplicities are real to swap the inner product. Unless $d\equiv  2 \pmod 4$ and $R_\alpha$ has $d/2$ rows, $R_\alpha$ is real and thus self-dual, meaning $R_\alpha\cong R_\alpha^{*}$ (Ref.~\cite{Alesker2010}, Lemma 3.2). In the exceptional case, if $R_\alpha$ is the self-dual (anti-self-dual) representation, then $R_\alpha^*$ is the anti-self-dual (self-dual) representation. These are distinct irreducible representations, but they share the same Young tableau. Since the arguments below only depend on properties of the Young tableau, we may replace $R_\alpha^*$ with $R_\alpha$ without loss of generality. Thus, $m$ is also the multiplicity of $R_I$ in $R_\alpha\otimes [2]$. Our goal is once again reframed: we must find choices of $I$ such that $R_I \not\subset R_\alpha \otimes [2]$. To make progress, we decompose the tensor product $R_\alpha \otimes [1] \otimes [1]$:
\begin{equation}
    R_I \subseteq \vcenter{\hbox{\begin{ytableau} \alpha \end{ytableau}}}\otimes \left( \vcenter{\hbox{\ydiagram{1}}} \otimes \vcenter{\hbox{\ydiagram{1}}} \right) 
= \vcenter{\hbox{\begin{ytableau} \alpha \end{ytableau}}} \otimes \left( \vcenter{\hbox{\ydiagram{1,1}}} \oplus  \vcenter{\hbox{\ydiagram{2}}} \oplus [0] \right)= \vcenter{\hbox{\begin{ytableau} \alpha \end{ytableau}}}\otimes \vcenter{\hbox{\ydiagram{1,1}}}  \oplus \vcenter{\hbox{\begin{ytableau} \alpha \end{ytableau}}}\otimes \vcenter{\hbox{\ydiagram{2}}} \oplus \vcenter{\hbox{\begin{ytableau} \alpha \end{ytableau}}}.
\end{equation}
We further note that if $R_\alpha\neq [0]$, then $R_\alpha\subseteq R_\alpha\otimes [1,1]$ because we can always add and subtract a box in the same position by our Newell-Littlewood rules for $\SON{d}$. The case where $R_\alpha =[0]$ is trivial and will be handled separately. To ensure $R_I \not\subset R_\alpha \otimes [2]$ while still satisfying $R_I \subseteq R_\alpha \otimes [1] \otimes [1]$, we must select an $R_I$ that appears in the decomposition of $R_\alpha \otimes [1,1]$ but does not appear in $R_\alpha \otimes [2]$. We determine this by comparing the Newell-Littlewood trace rules for both products. There are three cases for the contractions:
\begin{enumerate}
    \item \textbf{No Trace}: We apply standard $\SUN{d}$ rules. For $R_\alpha\otimes[1,1]$, we add two boxes to $R_\alpha$ such that they do not share a row, while for $R_\alpha\otimes [2]$ we add two boxes such that neither share the same column. The only configurations generated by $R_\alpha \otimes [1,1]$ that cannot be generated by $R_\alpha \otimes [2]$ are those where the two added boxes sit in the same column. Thus, $R_I$ can be formed by adding two boxes in the same column to $R_\alpha$. 
    
    \item \textbf{2 Traces}: Taking two traces requires us to remove two boxes from $R_\alpha$. For the product $R_\alpha \otimes [1,1]$, the antisymmetric trace rules imply that we can subtract any two boxes from $R_\alpha$ so long as they sit in different rows. For $R_\alpha \otimes [2]$, the symmetric trace rules enforce that we can subtract any two boxes as long as they sit in different columns. The unwanted intersection occurs when we subtract two boxes that share neither a row nor a column. The only terms that are unique to $R_\alpha \otimes [1,1]$ are thus formed by subtracting two boxes that occupy the exact same column of $R_\alpha$.
    
    \item \textbf{1 Trace}: Taking one trace leaves us with $(R_\alpha\setminus [1]) \otimes [1]$ for both diagrams. The resulting irreducible representations are identical for both products, meaning no $R_I$ can be extracted here. Note that this component naturally contains $R_\alpha$, because we can always subtract and add a box, so unless $R_\alpha=[0]$, $R_I\neq R_\alpha$.
\end{enumerate}

Overall, we must therefore generate $\{R_\alpha, R_I\}$ by adding or subtracting a box in the same column of $R_a$. An example is given in Figure~\ref{fig:possiblediagrams}. A caveat: in $\SON{2k}$, a tableau where the number of rows is exactly $k$ splits into self-dual and anti-self-dual parts. For instance, in $\SON{4}$, $[1,1]$ is reducible, and the above trace logic requires modification. To avoid this edge case, we assume $d \neq 4$. Thus, excluding the $\SON{4}$ and $R_\alpha=[0]$ edge cases, the only valid irreducible representations for $R_I$ are those formed by either adding or subtracting two boxes in a column of $R_\alpha$. Since $R_I\subseteq  [1]\otimes R_a$, we must specifically select this choice of $R_I$ from the former decomposition. This argument precedes modification rules for nonstandard tableaux, which may later force $[2] \subseteq R_\alpha \otimes R_I$ in certain $\SON{2k}$ situations. Our argument only states that it is \textit{necessary} that $R_I$ is generated by adding or subtracting two boxes in the trace step, not that it is sufficient.

Finally, we address the trivial case $R_\alpha=[0]$. Here, the decomposition reduces to $R_I\subseteq [1,1] \oplus [2] \oplus [0]$. Because we require $R_I \not\subset R_\alpha \otimes [2] = [0] \otimes [2] = [2]$, the valid choices for $I$ are simply the remaining components: $R_I \subseteq [1,1] \oplus [0]$.

\begin{figure}[t]
\centering
\begin{tikzpicture}[every node/.style={inner sep=2pt}]
  \matrix (M) [
      matrix of nodes,
      nodes={anchor=west},     
      column sep=16pt,
      row sep=10pt,
  ]{
    $R_a$ & Possible $\{R_\alpha,R_I\}$ (before modification rules) \\
    $\Rmu$     & $\left\{\,\PairAleft\,,\ \PairAright\,\right\}\quad\text{or}\quad
                 \left\{\,\PairBleft\,,\ \PairBright\,\right\}$ \\
  };
  \coordinate (sep) at ([xshift=-8pt]M-1-2.west);
  \draw[dashed] (sep |- M.north) -- (sep |- M.south);    
\end{tikzpicture}
\caption{Possible $R_\alpha$ and $R_I$ allowed by C1 for $R_a= [2^3,1] $ before modification rules. Crossed-out dashed red boxes indicate boxes that were deleted from $R_a$ and solid green boxes indicate boxes added to $R_a$.  }
\label{fig:possiblediagrams}
\end{figure}
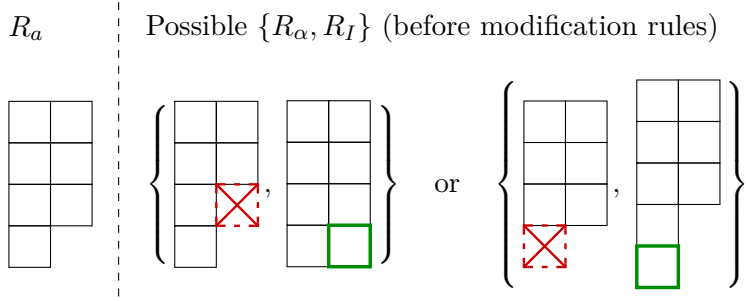

\section{Consistency of \texorpdfstring{$p$}{p}-Form Symmetries with Their Ideal Hydrodynamics}\label{Sec:Idealsolutions}

In this section, we demonstrate that $p$-form symmetries in $\SON{d}$ satisfy consistency conditions C1 and C2, and we derive their ideal constitutive relations. As a warm-up, we first rederive ideal non-relativistic magnetohydrodynamics (a one-form symmetry in $\SON{3}$) within our general framework, reproducing the results of Section~\ref{Sec:IdealMHD}. We then derive the constitutive relations for arbitrary $p$ in odd $d$ (we leave the verification of C2 to Appendix~\ref{Appdx:pformC2derivation}). The restriction to odd $d$ is mild and serves to exclude the self-dual case $p=d/2$ for even $d$, where we expect the same results to hold. For $p\neq d/2$ in even $d$, our calculation goes through unchanged. 

\subsection{Warm-up: One-Form Symmetry in \texorpdfstring{$\SON{3}$}{SO(3)}}
\label{Sec:oneformframework}

While already solved earlier, we show that our generic framework reproduces Section \ref{Sec:IdealMHD}. We consider the theory where $R_a=[1]$, $R_{\alpha}=[1]$, and $R_I=[0]$. Recall the conditions established in Section \ref{Sec:IdealMHD}: 
\begin{equation}
    \partial_t \rho_i +\epsilon_{ijk}\partial_j J_k=0, \quad f_k =\partial_k g, \quad  \partial_i \rho_i=0.
\end{equation}
In Appendix \ref{Appdx:enumeratingtheoriesSO3}, we demonstrate that such a theory satisfies our two consistency conditions and is thus compatible with hydrodynamics. To solve for the ideal hydrodynamic constitutive relations, we must evaluate the tensor $K^{(I)}_{ia,jb}$. Recall that $K^{(I)}_{ia,jb}$ will be proportional to the projector onto $R_\alpha=[1]$: 
\begin{equation}
    K^{(I)}_{ia,jb} =\frac{1}{2}(1-\kappa_I([1])) (\delta_{ij}\delta_{ab}-\delta_{ib}\delta_{ja}).
\end{equation}
By evaluating the Wigner $6j$ symbols (see Appendix \ref{Appdx:enumeratingtheoriesSO3}), we find that $\kappa_I(R_\alpha) =-1$. This yields 
\begin{align}
    \tau_{ji} &=P\delta_{ij} +g_jv_i - \rho_i\mu_j, \quad 
    J_{ai} =\rho_av_i-\rho_iv_a,
\end{align}
which agrees with Section \ref{Sec:IdealMHD}. 

\subsection{Generic \texorpdfstring{$p$}{p}-Form Symmetries in \texorpdfstring{$\SON{d}$}{SO(d)}}
\label{Sec:pformframewwork}

We now generalize to theories for $p\geq 1$ in $\SON{d}$ for odd $d$, where the density $\rho_a$ transforms in the $p$-form irreducible representation. We will further demand the current $J_{\alpha}$ transforms in the $(p+1)$-form irreducible representation. While we restrict to odd $d$ to avoid the subtleties of self-dual tableaux for even $d$ when $p=d/2$, the results in this section shall remain the same. Since $R_I$ must differ by two boxes in a column from $R_\alpha$, and $R_I\subseteq [1]\otimes R_a$, $g_I$ is constrained to transform in the $(p-1)$-form irreducible representation. By Hodge duality, this set-up will break into four distinct regimes of Young tableaux depending on the value of $p$:

\begin{center}
\begin{tabular}{| c | c | c | c | c |} 
  \hline
   & $R_a$ & $R_\alpha$ & $R_I$ & $R_{\tilde{\alpha}}$  \\ 
  \hline
  $p< \lfloor d/2\rfloor$ & $[1^p]$ & $[1^{p+1}]$ & $[1^{p-1}]$ & $[2,1^{p-1}]\oplus [1^{p-1}]$ \\ 
  \hline
  $p> \lceil d/2\rceil$ & $[1^{d-p}]$ & $[1^{d-p-1}]$ & $[1^{d-p+1}]$ & $[2,1^{d-p-1}]\oplus [1^{d-p+1}] $\\
  \hline
  $p= \lfloor d/2\rfloor$ & $[1^p]$ & $[1^{p}]$ & $[1^{p-1}]$ & $[2,1^{p-1}]\oplus [1^{p-1}]$\\
  \hline
  $p= \lceil d/2\rceil$ & $[1^{d-p}]$ & $[1^{d-p-1}]$ & $[1^{d-p}]$ & $[2,1^{d-p-1}]\oplus [1^{d-p}]$\\
  \hline
\end{tabular}
\end{center}

 To verify C2 across these four cases, we shall evaluate the generalized Clebsch-Gordan coefficients directly instead of computing the Wigner $6j$ symbols because they are challenging to compute in generic $\SON{d}$. We detail each case in Appendix \ref{Appdx:pformC2derivation} and find that the condition is universally satisfied. Furthermore, we demonstrate that $K^{(I)}_{ia,jb}=\delta_{ia}^{jb}$, which reduces the stress tensor and current to
\begin{equation}
    \tau_{ji} = P\delta_{ij} +g_j v_i - \mu_{[jI]}\rho_{[iI]}, \quad J_{ai} = (v\wedge \rho)_{ia},
\end{equation}
where $ I$ labels a contracted $(p-1)$-form index in the stress tensor. This confirms that our generalized method can systematically construct ideal hydrodynamics for arbitrary $p$-form symmetries.

\section{There is Likely Nothing Beyond \texorpdfstring{$p$}{p}-Form Symmetries in \texorpdfstring{$\SON{d}$}{SO(d)}}
\label{Sec:likelynothingelse}
In this section, we prove that, under mild restrictions detailed below, the only single-derivative ($m=1$) theories compatible with C1 and C2 are $p$-form symmetries. We first provide a physical picture for why single-derivative theories satisfying only C1 behave like one-form symmetries with additional structure. Such a picture provides physical intuition for why higher-form symmetries are special and motivates the no-go result in Section \ref{Sec:nothingelse}.
\subsection{Physical Motivation: Why Compatible Theories Look Like Higher-Form Symmetries}
We now present an example that illustrates how single derivative $(m=1)$ theories behave like one-form symmetries. We first focus on the one-form case in $\mathbb{R}^3$ before generalizing it to generic $m=1$ theories. Consider the finite volume $V$ to be a box $V=[x_-,x_+]\times [y_-,y_+]\times [z_-,z_+]$ in $\mathbb{R}^3$. The dipole symmetry $f_i= z\delta_{iz}$ is a valid multipolar charge, since $\partial_j \epsilon_{ijk}f_i=0$. We shall choose boundary conditions such that $\hat{n}_\perp \rho_{\perp}|_{\text{side walls}}=0$, where $\perp$ denotes the coordinates $\{x,y\}$ perpendicular to $\hat{z}$ (a repeated $\perp$ implies summation of indices in the set $\perp$), and the side walls are the faces of the box whose normal vectors are perpendicular to $\hat{z}$. The conserved multipolar charge is thus
\begin{equation}
    Q[f_i] = \int_{V}d^3x  \rho_zz\equiv  \int_{z_-}^{z_+} dz \Psi(z)z,
\end{equation}
where we define the planar flux $\Psi(z) \coloneq \int_{\mathrm{plane}} dxdy \rho_z$. The boundary condition is crucial because it makes $\Psi(z)$ independent of $z$:
\begin{equation}
    \partial_z \Psi(z) = \int_{\mathrm{plane}} dxdy \partial_z \rho_z = - \int_{\mathrm{plane}} dxdy (\partial_x \rho_x+\partial_y\rho_y) = -\int_{\text{side walls}} dl \hat{n}_\perp \rho_{\perp} =0 .
\end{equation}
Momentum conservation follows:
\begin{equation}
    0 = \{P_z,Q[f_i;V]\} = \int_V d^3x \partial_z\rho_z z= \int_{z_-}^{z_+} dz z\partial_z \Psi(z) = [z\Psi(z)]_{z_-}^{z_+} - \int_{z_-}^{z_+} dz  \Psi(z)=0,
\end{equation}
where the last expression follows from the $z$-independence of $\Psi(z)$. By rotational invariance, the same argument holds for any direction $\hat{n}$, not just $\hat{z}$. Momentum conservation in a finite volume is thus compatible with the conservation of a flux of $\rho_z$ through codimension-one surfaces in any direction (which is the geometric picture of a one-form symmetry).

The same construction follows for any single-derivative ($m=1$) theory. We work in $d$ spatial dimensions with our volume $V$ a $d$-dimensional box. We denote the $x_d$ direction $\hat{z}$,  choose $f_a=c_I\Gamma_{zaI}z$ ($c_I$ is some constant that transforms in $R_I$), choose boundary conditions such that $\Gamma_{\perp aI}\hat{n}_\perp \rho_{a}|_{\text{side walls}}=0$, ($\perp$ now denotes the coordinates $\{x_1,\cdots ,x_{d-1}\}$ perpendicular to $\hat{z}$ and the side walls are the boundary of the $d$-dimensional box with normal vectors perpendicular to $\hat{z}$). Defining the generalized flux
\begin{equation}
    \Psi(z) = \int_{\perp}d^{d-1}x \Gamma_{zaI}\rho_a,
\end{equation}
and carrying out identical arguments to the one-form case shows that $\partial_z\Psi(z)=0$ and thus $\{P_z,Q[f_a;V]\}=0$. For each fixed $I$, and (by rotational invariance) in any direction, the vector $\Gamma_{zaI}\rho_a$ is divergence-free, so its flux through codimension-one surfaces is conserved. The generalized Gauss's law constraint generates the same one-form-like fluxes, now built from $\Gamma_{zaI}\rho_a$ instead of $\rho_z$. 

We emphasize that our argument shows every compatible theory is at \emph{least} a one-form symmetry. In general, however, the extra $I$ index and the contraction with the $a$ index give a theory more structure than a one-form symmetry. A $p$-form symmetry is the case in which this extra structure remains consistent with momentum-conserving hydrodynamics. Otherwise, it is precisely this additional structure that violates condition C2, thus obstructing momentum conservation and generating a Goldstone mode. This is our no-go result in Section \ref{Sec:nothingelse}.

\subsection{Proof of the No-Go Result }
\label{Sec:nothingelse}

In this section, we show that, given a sufficiently large odd dimension $d$ and $R_\alpha$ irreducible, the only theories in $\SON{d}$ compatible with Consistency Conditions C1 and C2 are $p$-form symmetries. Namely, we prove the following theorem:
\begin{thm}
\label{thm:pform}
   Let our group $G=\SON{d}$, for odd $d$. Let $R_a$, $R_\alpha$, and $R_I$ be standard irreducible representations that satisfy consistency conditions C1 and C2. Further, assume that $R_a$ is chosen such that $d\geq 4((\lambda_a)'_1+1)$. Then for some $p\geq 1$, $R_a=[1^p]$ and either $(R_\alpha, R_I) =([1^{p+1}],[1^{p-1}])$ or  $(R_\alpha, R_I) =([1^{p-1}],[1^{p+1}])$.
\end{thm}
 The condition $d\geq 4((\lambda_a)'_1+1)$ places us in a regime in which $R_I\otimes R_\alpha$ is small enough such that its irrep decomposition requires no modification rules. This is because we will not have to apply modification rules in the regime where $(\lambda_I)_1' +(\lambda_\alpha)_1' \leq \lfloor d/2\rfloor$. Since $(\lambda_I)_1'$ and $(\lambda_\alpha)_1'$ are bounded by $(\lambda_a)_1'+1$, we obtain the provided bound. We believe this condition can be relaxed and the theorem will still hold, however it may require significant casework to verify. Likewise, we further believe that the theorem will hold for even $d$ (indeed, most of the proof of Theorem~\ref{thm:pform} does not require odd $d$), however we do not believe it shall be fruitful to deal with the self-dual edge cases with certain Young tableaux being further reducible.  We shall prove Theorem~\ref{thm:pform} in Appendix \ref{Appdx:proofmaintheorem}. We provide an outline of the proof below.
\begin{proof}[Outline of Proof]
We first observe that the tensor $\hat{M}\equiv M_{ia,jb}\coloneq \Gamma_{ibI}\Gamma_{jaI}$ is itself an intertwiner in $\mathrm{End}_{\SON{d}}([1]\otimes R_a)$. This algebra is generated by the interaction operator $\hat{X}_{\mathrm{int}}$, which is defined in Appendix~\ref{Appdx:proofmaintheorem}. The proof shall involve deriving explicit eigenvalues for $\hat{X}_{\mathrm{int}}$ (which, in turn, will require finding the explicit quadratic Casimir eigenvalues). We utilize these analytic forms to derive several trace identities for various products of $\hat{X}_{\mathrm{int}}$ and $\hat{M}$. These trace identities will provide a system of equations in terms of various dimensions of irreducible representations and eigenvalues of $\hat{X}_{\mathrm{int}}$. We further assume consistency conditions C1 and C2 to reduce these systems of equations to a single consistency equation that we label $C^*$. Finally, we show that $C^*$ cannot hold for non-column irreducible representations when $d$ is above $(\lambda_a)_1'+(\lambda_a)_2'$. 
\end{proof}

\section{Concluding Thoughts}
In this paper, we studied which exotic conservation laws a momentum-conserving system can have without Goldstone modes. We find that, under the assumptions laid out in Table~\ref{tab:assumptions}, such theories must obey a generalized Gauss's law condition, since local sources and sinks would obstruct translational invariance, and that any net global charge must lie on the boundary of the system. Generically, these theories must be at least one-form symmetries, with conserved fluxes through surfaces, as in magnetohydrodynamics. We argued that only theories with a single-derivative Gauss's law constraint are compatible with momentum conservation, and thus one cannot formulate a consistent gauged fracton magnetohydrodynamics with momentum conservation, at least in theories with purely scalar multipole charges.  We further established two purely representation theoretic consistency conditions (C1 and C2) that highly constrain the class of such theories, which we find to be essentially only $p$-form symmetries. This has a strong physical consequence: momentum-conserving theories that are not $p$-form symmetries must either violate an assumption in Table~\ref{tab:assumptions} or produce a Goldstone mode. While we only prove this for irreducible current representations $R_\alpha$, odd spatial dimensions $d$, and $R_a$ sufficiently small, we expect the result to hold generically. 

Two of these restrictions seem to be technical rather than essential. Very little of our proof should change for even $d$ (as only exceptional representations behave differently). The size constraint on $R_a$ is imposed to avoid modification-rule subtleties in the tensor-product decomposition and avoid case-work. We are unaware of any case that would adjust our results and believe it could be removed with a more sophisticated proof. 

The reducibility of $R_\alpha$ is the only condition that may be exceptional. For reducible $R_\alpha$, condition C1 becomes much more stringent, as the irreducible representation $R_I$ must satisfy it for each irreducible representation in $R_\alpha$. The condition C2, however, becomes less stringent, as fewer representations remain in the complementary representation $R_{\tilde{\alpha}}$. In principle, an exotic $R_\alpha$ could exist that perfectly balances the tightening of C1 and the loosening of C2, evading our no-go result. We have found no such example and suspect none exists: the generalized Gauss's law constraint already makes our theory behave like a $p$-form symmetry. We leave such a confirmation (or a counterexample) to future work.

Our framework also rests on a potentially strong assumption about the allowed multipolar charges $Q[f_a]= \int_V d^dx f_a \rho_a $; namely that every solution of $\Gamma_{ia\alpha} \partial_i f_a=0$ is exact and of the form $f_a= \Gamma_{jaI}\partial_j g_I$ (for single-derivative theories). Even if we grant that $p$-form symmetries are the only theories allowed after relaxing the above three assumptions, it remains possible that this is an artifact of exactness forcing theories to resemble $p$-forms. Whether theories admitting non-exact $f_a$ can be compatible with momentum conservation is the most interesting open direction.

Additionally, our analysis is restricted to scalar multipolar charges $Q[f_a]$ and does not address the case in which it transforms in some representation: i.e. $Q_{b}[f_a] = \int_V d^dx f_a \rho_b  $, for $(f_a,\rho_b)$ transforming in representations $(R_a,R_b)$. This is precisely the structure for dipole-conserving fractons, in which $f_a=x_a$ transforms in the vector representation and $\rho_b=\rho$ transforms as a scalar. Extending our framework to such charges provides a natural direction towards a momentum-conserving framework for generalized fracton-hydrodynamics. Such exotic charges may yield different constraints and allow a broader class of theories than a scalar $Q[f_a]$. 

Finally, our analysis was restricted to $G=\SON{d}$, but our framework works for \textit{any} finite-dimensional group. It will be interesting to see whether other groups allow for exotic charges beyond $p$-forms. 

\section*{Acknowledgements}
Claude Fable 5 helped prove Theorem~\ref{thm:pform}; all steps were independently verified by us. Claude Opus 5 and 5.5 helped create the Young diagrams and Table~\ref{tab:assumptions}. ChatGPT 5.6 gave helpful feedback on LaTeX style. All other results are our own. This work was supported by the U.S. Department of Energy, Office of Science, Basic
Energy Sciences under Award No. DE-SC0021346 (BB, RN), and by the National Science Foundation under CAREER Grant DMR-2145544 (AL).

\begin{appendix}

\makeatletter


\let\@sectioncntformat\@seccntformat

\let\@hangfrom@section\@hang@from


\renewcommand\section{\@ifstar{\AndySectionStar}{\@startsection{section}{1}{\z@}%

  {-3.5ex \@plus -1ex \@minus -.2ex}%

  {2.3ex \@plus .2ex}%

  {\normalfont\Large\bfseries}}}

\renewcommand\subsection{\@ifstar{\AndySubsectionStar}{\@startsection{subsection}{2}{\z@}%

  {-3.25ex\@plus -1ex \@minus -.2ex}%

  {1.5ex \@plus .2ex}%

  {\normalfont\large\bfseries}}}

\renewcommand\subsubsection{\@ifstar{\AndySubsubsectionStar}{\@startsection{subsubsection}{3}{\z@}%

  {-3.25ex\@plus -1ex \@minus -.2ex}%

  {1.5ex \@plus .2ex}%

  {\normalfont\normalsize\bfseries}}}

\makeatother

\renewcommand{\thesubsection}{\thesection.\arabic{subsection}}
\renewcommand{\theequation}{\thesection.\arabic{equation}}
\renewcommand{\thesubsubsection}{\thesubsection.\arabic{subsubsection}}
\makeatletter
\renewcommand\section{\@ifstar{\AndySectionStar}{\@startsection {section}{1}{-\parindent}%
  {-3.5ex \@plus -1ex \@minus -.2ex}%
  {2.3ex \@plus .2ex}%
  {\normalfont\Large\bfseries}}}
\renewcommand\subsection{\@ifstar{\AndySubsectionStar}{\@startsection{subsection}{2}{-\parindent}%
  {-3.25ex\@plus -1ex \@minus -.2ex}%
  {1.5ex \@plus .2ex}%
  {\normalfont\large\bfseries}}}
\renewcommand\subsubsection{\@ifstar{\AndySubsubsectionStar}{\@startsection{subsubsection}{3}{-\parindent}%
  {-3.25ex\@plus -1ex \@minus -.2ex}%
  {1.5ex \@plus .2ex}%
  {\normalfont\normalsize\bfseries}}}
\makeatother

\section{Enumerating All Possible Theories in \texorpdfstring{$\SON{3}$}{SO(3)}}
\label{Appdx:enumeratingtheoriesSO3}
In this section, we systematically enumerate all possible theories compatible with momentum conservation in $\SON{3}$ and find that the only such theories are $p$-form symmetries. Our strategy is to apply the First Consistency Condition (C1) first, which eliminates almost all candidate theories, and then verify the Second Consistency Condition (C2) for the remaining candidates using the Wigner 6$j$ symbols enumerated in Section \ref{Sec:6jtable}.
\subsection{Theory Elimination via C1 }
Recall that C1 requires $R_\alpha \otimes R_I$ to contain neither $[2]$ nor $[0]$. Further, recall that 6$j$ symbols are nonzero only if they satisfy the triangle inequality: for any triplet $(x,y,z)$ we must have that $|x-y|\leq z\leq x+y$. Because the vector representation has $j_V=1$ in $\SON{3}$, the triangle inequality for the triplets $(j_{\tilde{\alpha}},1,j_a)$ and $ (j_I,1,j_a)$ forces 
\begin{equation}
    j_{\tilde{\alpha}},j_I\in \{j_a-1,j_a,j_a+1\},
\end{equation}
leaving nine candidate theories for a given $j_a$. We look at three cases of $j_a$:
\cse{$j_a\geq 2$: no compatible theories.} The available representations for $R_\alpha$ and  $R_I$ are $\{[j_a-1],[j_a], [j_a+1]\}$. We can eliminate all nine pairings:
\begin{itemize}
    \item \emph{$R_\alpha=R_I$}: The decomposition for $R_\alpha\otimes R_I$ will always contain $[0]$, violating C1.
    \item \emph{$R_\alpha$ and $R_I$ differ by two boxes} (i.e. one is $[j_a+1]$ and one is $[j_a-1]$): applying the trace subtraction $\min(j_\alpha, j_I)$ times to the decomposition of $R_\alpha\otimes R_I$ will leave $[2]$ in the resulting sum.
    \item \emph{$R_\alpha$ and $R_I$ differ by one box}: The decomposition of the tensor products
    \begin{equation}
    [j_a]\otimes[j_a-1]=\bigoplus_{n=1}^{2j_a-1}[n],
\qquad
[j_a]\otimes[j_a+1]=\bigoplus_{n=1}^{2j_a+1}[n]
\end{equation}
both contain $[2]$. 
\end{itemize}
This exhausts all nine possible theories, so $j_a\geq 2$ is incompatible with momentum conservation in $\SON{3}$ by just C1.

\cse{$j_a=1$: two compatible theories.} The valid candidates are $R_\alpha, R_I \in \{[0], [1]\}$  (neither $R_I$ nor $ R_\alpha =[2]$ because then $[2]\subseteq R_\alpha \otimes R_I$). Similarly, we cannot have $R_\alpha =R_I$, because $[0]\subseteq R_\alpha \otimes R_I$. This leaves only two viable possibilities for $j_a=1$:
\begin{itemize}
    \item $(R_\alpha, R_I) = ([1],[0])$, corresponding to one-form hydrodynamics,
    \item $(R_\alpha, R_I) = ([0],[1])$, corresponding to two-form hydrodynamics.
\end{itemize}
We shall verify that both theories satisfy C2 in the next section. 

\cse{$j_a=0$: conventional hydrodynamics.} Here, $R_a=[0]$, leaving $R_\alpha,R_I \subseteq [1]\otimes R_a=[1]$. It therefore must be the case that $R_\alpha =R_I=[1]$ and $R_{\tilde{\alpha}}$ is empty, causing an apparent inconsistency. A scalar, rotationally invariant charge density $\rho$ is known to be consistent with momentum conservation in conventional hydrodynamics. The resolution is twofold. First, in such theories $f_a=f(x)$ is a scalar, so the derivative constraint forces $f_a$ to be constant. The charge $Q[f] = f\int_V d^dx \rho $ is thus trivially compatible with momentum conservation and C1 is not applicable. Second, the Lagrange multiplier $\lambda_I$ and the postulated coefficient $A_{ai}$ both vanish in such theories (i.e. $s_i=Pv_i$), so C2 never needs to be explicitly checked. 

\subsection{Verification of C2 for the Remaining Theories}
For the two remaining theories, we confirm that C2 holds (i.e., that $\kappa_I(R_{\tilde{\alpha}})$ is constant across $R_{\tilde{\alpha}}$). Reading off the relevant entries of the table in Section \ref{Sec:6jtable}:
\begin{itemize}
    \item \emph{One-form} $(R_\alpha,R_I)=([1],[0])$: here $\Gamma_{jaI}=\delta_{ja}$ and so $\lVert \Gamma_{jaI}\rVert^2=3 $. Further, $R_{\tilde{\alpha}}=[0]\oplus [2]$. Reading off our table for Cases 1 and 7, we see that $\kappa_I(R_{\tilde{\alpha}})=1$ in both cases, so C2 is satisfied.
    \item  \emph{Two-form} $(R_\alpha,R_I)=([0],[1])$: here $\Gamma_{jaI}=\epsilon_{jaI}$ and so $\lVert \Gamma_{jaI}\rVert^2=6 $. Further, $R_{\tilde{\alpha}}=[1]\oplus [2]$. Reading off our table for Cases 5 and 8, we see that $\kappa_I(R_{\tilde{\alpha}})=1$ in both cases, so C2 is satisfied.
\end{itemize}
Both the one- and two-form cases satisfy both consistency conditions, and no other theory in $\SON{3}$ does. The only theories in $\SON{3}$ compatible with momentum conservation are thus one- and two-form symmetries.

\subsection{Table of \texorpdfstring{$6j$}{6j} Symbols}
\label{Sec:6jtable}
For completeness, we list all nine possible theories in this section. Because the fundamental vector representation in this group has spin $j_V=1$, our expression for $\kappa_I(R_{\tilde{\alpha}})$ takes the form
\begin{equation}
        \kappa_I(R_{\tilde{\alpha}})=\lVert \Gamma_{jaI} \rVert^2 \begin{Bmatrix} j_{\tilde{\alpha}} & 1 & j_a \\ j_{I} & 1 & j_a \end{Bmatrix}=\lVert \Gamma_{jaI} \rVert^2\begin{Bmatrix} j_a & j_{\tilde{\alpha}} & 1  \\ j_a & j_{I} & 1\end{Bmatrix}.
\end{equation}
Recall that the $j_V=1$ constraint causes $j_{\tilde{\alpha}},j_I\in \{j_a-1,j_a,j_a+1\} $, restricting the number of allowed theories to nine. Note the inherent symmetry of the 6$j$-symbols:
\begin{equation}
    \begin{Bmatrix} j_a & j_{\tilde{\alpha}} & 1  \\ j_a & j_{I} & 1\end{Bmatrix}= \begin{Bmatrix} j_a & j_{I} & 1  \\ j_a & j_{\tilde{\alpha}} & 1\end{Bmatrix}.
\end{equation}
Because the expression is invariant under the exchange of $j_{\tilde{\alpha}}$ and $j_I$, several cases become identical. Using Table 9.2 in Ref.~\cite{Varshalovich1988}, we find
\begin{enumerate}
    \item $j_{\tilde{\alpha}},j_I=(j_a-1,j_a-1)$: \begin{equation}
         \kappa_I(R_{\tilde{\alpha}})=\lVert \Gamma_{jaI} \rVert^2 \begin{Bmatrix} 1 & j_a -1 & j_a  \\ 1 & j_a -1 & j_a\end{Bmatrix}= \frac{\lVert \Gamma_{jaI} \rVert^2}{j_a(2j_a-1)(2j_a+1)}
    \end{equation}
    \item $j_{\tilde{\alpha}},j_I=(j_a-1,j_a)$:
     \begin{equation}
         \kappa_I(R_{\tilde{\alpha}})=\lVert \Gamma_{jaI} \rVert^2 \begin{Bmatrix}   j_a  & j_a & 1  \\   j_a-1  & j_a & 1\end{Bmatrix} =-\frac{\lVert \Gamma_{jaI} \rVert^2}{j_a(2j_a+1)}
    \end{equation}
    \item $j_{\tilde{\alpha}},j_I=(j_a-1,j_a+1)$:
    \begin{equation}
          \kappa_I(R_{\tilde{\alpha}})= \lVert \Gamma_{jaI} \rVert^2  \begin{Bmatrix}   j_a+1  & j_a & 1  \\   j_a-1  & j_a & 1\end{Bmatrix} = \frac{\lVert \Gamma_{jaI} \rVert^2}{2j_a + 1}
    \end{equation}
    \item $j_{\tilde{\alpha}},j_I=(j_a,j_a-1)$: This is identical to Case 2 by symmetry 
    \item $j_{\tilde{\alpha}},j_I=(j_a,j_a)$:
    \begin{equation}
         \kappa_I(R_{\tilde{\alpha}})=\lVert \Gamma_{jaI} \rVert^2 \begin{Bmatrix}   j_a  & j_a & 1  \\   j_a  & j_a & 1\end{Bmatrix} = \lVert \Gamma_{jaI} \rVert^2 \frac{j_a(j_a+1)-1}{j_a(2j_a+1)(j_a+1)}
    \end{equation}
    
    \item $j_{\tilde{\alpha}},j_I=(j_a,j_a+1)$:
       \begin{equation}
          \kappa_I(R_{\tilde{\alpha}})= \lVert \Gamma_{jaI} \rVert^2\begin{Bmatrix}   j_a+1  & j_a & 1  \\   j_a  & j_a & 1\end{Bmatrix} =\frac{\lVert \Gamma_{jaI} \rVert^2 }{(j_a+1)(2j_a+1)} 
    \end{equation}
    \item $j_{\tilde{\alpha}},j_I=(j_a+1,j_a-1)$: This is identical to Case 3 by symmetry 
    \item $j_{\tilde{\alpha}},j_I=(j_a+1,j_a)$: This is identical to Case 6 by symmetry 
    \item $j_{\tilde{\alpha}},j_I=(j_a+1,j_a+1)$:
     \begin{equation}
           \kappa_I(R_{\tilde{\alpha}})=\lVert \Gamma_{jaI} \rVert^2   \begin{Bmatrix}   j_a+1  & j_a & 1  \\   j_a+1  & j_a & 1\end{Bmatrix} =\frac{\lVert \Gamma_{jaI} \rVert^2 }{(2j_a+1)(j_a+1)(2j_a+3)}
    \end{equation}
\end{enumerate}

\section{Derivation of C2 for \texorpdfstring{$p$}{p}-form Symmetries }
\label{Appdx:pformC2derivation}
In this Appendix, we evaluate the four index-contraction cases outlined in Section \ref{Sec:pformframewwork}. We remind the reader that our index contraction convention for $p$-forms sums over unique, ordered sets of indices. Consequently, the generalized Kronecker delta and Levi-Civita symbols utilized throughout these derivations lack conventional normalization constants. For instance, in our convention, if $a$ and $b$ are both antisymmetric $p$-forms, then $\delta_{a}^{b}A_b=A_{a}$, where the traditional $p!$ combinatorial prefactor associated with unrestricted index sums is omitted. 

\subsection{Case 1: \texorpdfstring{$p <\lfloor d/2\rfloor$}{p< ⌊d/2⌋} }

We begin by defining the generalized Clebsch-Gordan coefficients for our spaces: $\Gamma_{jaI } = \delta_{jI}^{a} $.  We can apply a Laplace expansion to these delta functions, yielding $ \delta_{jI}^{a}= \sum_{n}(-1)^{n+1} \delta_{j}^{a_n} \delta_{I}^{a \setminus\{a_{n}\}}  $. Expanding the contraction gives
\begin{align}
     \Gamma_{ja I } \Gamma_{ia\tilde{\alpha}}\Gamma_{ibI}&=\delta_{jI}^{a} \delta_{iI }^{b } \Gamma_{ia\tilde{\alpha}} = \sum_{n,m}(-1)^{n+m} \delta_{j}^{a_n}\delta_{i}^{b_m} \delta_{I}^{a \setminus \{a_{n}\}} \delta_{I}^{b \setminus\{b_{m}\}}\Gamma_{ia \tilde{\alpha}} = \sum_{n,m}(-1)^{n+m} \delta_{j}^{a_n} \delta_{b \setminus \{b_{m}\}}^{a \setminus \{a_{n}\}}\Gamma_{b_ma\tilde{\alpha}} \nonumber \\
     &= -\sum_{m}(-1)^{m}  \delta_{j,b\setminus \{b_m\}}^{a}\Gamma_{b_m a \tilde{\alpha} } = -\sum_{m}(-1)^{m}  \Gamma_{b_m [j,b\setminus \{b_m\}]\tilde{\alpha}}.
\end{align}
We must now evaluate the two possibilities for $R_{\tilde{\alpha}}$. For the totally antisymmetric case $R_{\tilde{\alpha}}=[1^{p-1}]$, we have $\Gamma_{b_m [j,b \setminus\{b_m \}]\tilde{\alpha} }= \delta_{b_m, \tilde{\alpha }}^{j,b \setminus \{b_m\}}$. Let us define the composite index $B^m\coloneq [j,b\setminus \{b_m \}]$ and perform another Laplace expansion along $b_m$:
\begin{equation}
   \delta_{b_m,\tilde{\alpha}}^{j,b \setminus \{b_m\}}=\sum_{n} (-1)^{n+1} \delta_{b_m }^{B^m_n} \delta_{\tilde{\alpha}}^{B^m\setminus \{B^m_n\}}.
\end{equation}
Since $b$ is a totally antisymmetric $p$-form, $\delta_{b_m }^{B^m_n}=0$ unless $B^m_n=j$. Therefore, the expansion becomes $ \delta_{b_m,\tilde{\alpha}}^{B^m}=\delta_{b_m }^{j} \delta_{\tilde{\alpha}}^{b\setminus \{b_m\}} $. Substituting this back yields
\begin{equation}
   \Gamma_{ja I } \Gamma_{ia\tilde{\alpha}}\Gamma_{ibI}= -\sum_{m}(-1)^{m} \delta_{b_m }^{j} \delta_{\tilde{\alpha}}^{b\setminus \{b_m\}} =  \delta_{j\tilde{\alpha}}^{b } = \Gamma_{jb\tilde{\alpha} }.
\end{equation}
Thus, $\kappa_I([1^{p-1}])=1$ for the totally antisymmetric representation. For the hook irreducible representation $R_{\tilde{\alpha}}=[2,1^{p-1}]$, the  antisymmetrization over $[j,b]$ in  $\Gamma_{b_m[j,b \setminus \{b_m \}]\tilde{\alpha} }$ vanishes. This occurs because the hook Young tableau contains a symmetric index pair; projecting this entirely into the completely antisymmetric $[1^{p+1}]$ representation annihilates the tensor by Schur's lemma. Therefore
\begin{equation}
   \Gamma_{jb\tilde{\alpha}}+ \sum_{m}(-1)^{m}  \Gamma_{b_m[j,b\setminus \{b_m\}]\tilde{\alpha}}=0,
\end{equation}
which matches $\Gamma_{ja I } \Gamma_{ia\tilde{\alpha}}\Gamma_{ibI} = \Gamma_{jb\tilde{\alpha}}$, thus establishing that $\kappa_I([2,1^{p-1}]) =1$ in the hook case as well. We now evaluate the projector $ P_{ia,jb}^{(\alpha)} = (\text{dim}(R_\alpha)/\lVert \Gamma_{ia\alpha}\rVert^2)\Gamma_{ia\alpha}\Gamma_{jb\alpha}$. By simple index contractions, we have $\Gamma_{ia\alpha}\Gamma_{jb\alpha} = \delta_{ia}^{\alpha}\delta_{jb}^{\alpha}= \delta_{ia}^{jb}$. Further, because $R_\alpha$ is a $(p+1)$-form, its dimension is $\text{dim}(R_\alpha)=\binom{d}{p+1}$. The norm $\lVert \Gamma_{ia\alpha}\rVert^2 = \delta_{ia}^{ia}$ requires summing over all valid, non-overlapping index choices. There are $\binom{d}{p}$ ways to choose $a$, leaving $(d-p)$ remaining choices for $i$, producing $(d-p)\binom{d}{p} = (p+1)\binom{d}{p+1}$. Finally, to solve $\kappa_I(R_\alpha)$, we return to our earlier sum, replacing $\tilde{\alpha}\to \alpha$:
\begin{equation}
    \Gamma_{jaI}\Gamma_{ibI}\Gamma_{ia\alpha} = -\sum_{m}(-1)^m \Gamma_{b_m [j,b\setminus \{b_m\}]\alpha } = -\sum_{m}(-1)^m \delta_{b_m,[j,b\setminus \{b_m\}]}^{\alpha} =-p\delta_{[j,b]}^{\alpha}=-p\Gamma_{jb\alpha},
\end{equation}
where the sum evaluates to $-p$ because shifting the index $b_m$ to the $m$-th slot requires $m$ swaps, generating a $(-1)^m$ phase that cancels the sum's alternating sign, repeated over $p$ identical terms. Thus, $K^{(I)}_{ia,jb}=(1-\kappa_I(R_\alpha))P^{(\alpha)}_{ia,jb}=\delta_{ia}^{jb} $ and $J_{ai}=(v\wedge\rho)_{ia}$. Henceforth, we shall drop the proportionality constants whenever we evaluate $K^{(I)}_{ia,jb}$, since the following three cases will end up identical to the first case. 

\subsection{Case 2: \texorpdfstring{$p>\lceil d/2\rceil$}{p>⌈d/2⌉} }
In this regime, we work in the Hodge dual space. Our coefficients become $\Gamma_{jaI}= \epsilon_{a}^{a'}\epsilon_{I}^{I'}\delta_{ja'}^{I'}$, where primes denote the dualized indices ($a'$ is a $(d-p)$-form and $I'$ is an $(d-p+1)$-form, respectively). Simplifying the relevant contraction gives 

\begin{align}
    \Gamma_{jaI}\Gamma_{ia\tilde{\alpha}}\Gamma_{ibI} &= (\epsilon_{a}^{a'}\epsilon_{I}^{I_1'}\delta_{ja'}^{I_1'} )(\epsilon_{b}^{b'}\epsilon_{I}^{I'_2}\delta_{ib'}^{I_2'})\Gamma_{ia\tilde{\alpha}}= \epsilon_{a}^{a'}\epsilon_{b}^{b'}\delta_{I_1'}^{I_2'} \delta_{ja'}^{I_1'} \delta_{ib'}^{I_2'}\Gamma_{ia\tilde{\alpha}}\nonumber
    \\&=\epsilon_{a}^{a'}\epsilon_{b}^{b'} \delta_{ib'}^{ja'}\Gamma_{ia\tilde{\alpha}}= \epsilon_{a}^{a'}\epsilon_{b}^{b'} \epsilon_{ib'}^I\epsilon_{ja'}^I\Gamma_{ia\tilde{\alpha}} =\delta_{iI}^b\delta_{jI}^a \Gamma_{ia\tilde{\alpha}}.
\end{align}
 Since $I$ is a $(p-1)$-form and $a$ and $b$ are $p$-forms, this expression is exactly the same as Case 1. The projector $K^{(I)}_{ia,jb}$ is further evaluated as
 \begin{align}
     K^{(I)}_{ia,jb}&= (1-\kappa_I(R_\alpha)) P_{ia,jb}^{(\alpha)}\propto \Gamma_{ia\alpha}\Gamma_{jb\alpha} = \epsilon_{a}^{a'}\epsilon_{b}^{b'}\epsilon_{\alpha}^{\alpha_1'}\epsilon_{\alpha}^{\alpha_2'}\delta_{i\alpha_1'}^{a'}\delta_{j\alpha_2'}^{b'} \nonumber \\&=\delta_{\alpha_1'}^{\alpha_2'}\epsilon_{a}^{a'}\epsilon_{b}^{b'}\delta_{i\alpha_1'}^{a'}\delta_{j\alpha_2'}^{b'} =\epsilon_{a}^{a'}\epsilon_{b}^{b'}\delta_{i\alpha_1'}^{a'}\delta_{j\alpha_1'}^{b'}= \epsilon_{a}^{i\alpha_1'}\epsilon_{b}^{j\alpha_1'}=\delta_{ia}^{jb}.
 \end{align}
 As before, this establishes that $K^{(I)}_{ia,jb}= \delta_{ia}^{jb}$.
 
\subsection{Case 3: \texorpdfstring{$p=\lfloor d/2\rfloor$}{⌊d/2⌋}  }
In this regime, the representations $R_{\tilde{\alpha}}$ and $R_I$ are identical to Case 1, meaning the generalized Clebsch-Gordan coefficients evaluate trivially to the exact same results. Only the representation $R_{\alpha}$ differs. Consequently, the projector $P_{ia,jb}^{(\alpha)}\propto \Gamma_{ia\alpha}\Gamma_{jb\alpha}$ must be evaluated differently. Here, we must use the dualized variables, defining $\Gamma_{ia\alpha}= \epsilon_{a}^{a'}\epsilon_{\alpha}^{\alpha'}\delta_{i\alpha'}^{a'}$, where the primed indices again denote the Hodge dual forms. Evaluating the projector yields
\begin{align}
    P_{ia,jb}^{(\alpha)} &\propto \Gamma_{ia\alpha}\Gamma_{jb\alpha} = (\epsilon_{a}^{a'}\epsilon_{\alpha}^{\alpha'_1}\delta_{i\alpha'_1}^{a'} )(\epsilon_{b}^{b'}\epsilon_{\alpha}^{\alpha_2'}\delta_{j\alpha_2'}^{b'})=\delta_{\alpha'_1}^{\alpha_2'}\epsilon_{b}^{b'}\epsilon_{a}^{a'}\delta_{i\alpha_1'}^{a'} \delta_{j\alpha_2'}^{b'} \nonumber\\
    &= \epsilon_{b}^{b'}\epsilon_{a}^{a'}\delta_{i\alpha_1'}^{a'} \delta_{j\alpha_1'}^{b'} = \epsilon_{a}^{i\alpha_1'} \epsilon_{b}^{j\alpha_1'} =\delta_{ia}^{jb}.
\end{align}
As before, this establishes that $K^{(I)}_{ia,jb}= \delta_{ia}^{jb}$.

\subsection{Case 4: \texorpdfstring{$p=\lceil d/2\rceil$}{p = ⌈d/2⌉}}

For this final case, the generalized Clebsch-Gordan coefficient is given by $\Gamma_{jaI}= \epsilon_{a}^{a'}\epsilon_{I}^{I'}\delta_{ja'}^{I'}$, where the primed variables again represent the dualized variables defined by the corresponding Young tableaux. Following the exact reasoning established in Case 2, we contract the coefficients to find
\begin{equation}
    \Gamma_{jaI}\Gamma_{ia\tilde{\alpha}}\Gamma_{ibI} = \delta_{iI}^{b}\delta_{jI}^a\Gamma_{ia\tilde{\alpha}}.
\end{equation}
This is almost identical to Case 2.  When the representation is the hook $R_{\tilde{\alpha}}=[2,1^{d-p-1}]$, this evaluates exactly as it did in Case 2. However, it requires a slight modification when $R_{\tilde{\alpha}}=[1^{d-p}]$. To resolve this, we expand $\Gamma_{ia\tilde{\alpha}}$ explicitly:
\begin{equation}
\Gamma_{ia\tilde{\alpha}}=\epsilon_{a}^{a'}\epsilon_{\tilde{\alpha}}^{\tilde{\alpha}'}\delta_{\tilde{\alpha}'}^{ia'} = \epsilon_{a}^{a'}\epsilon_{\tilde{\alpha}}^{ia'}=\pm \delta_{i\tilde{\alpha}}^{a}.
\end{equation}
Consequently, the dualized Clebsch-Gordan coefficient will reduce to the non-dualized coefficients (up to an unimportant sign), as in Case 2. Finally, to evaluate our projector $K^{(I)}_{ia,jb}$, we observe that the projector $P_{ia,jb}\propto \Gamma_{ia\alpha}\Gamma_{jb\alpha}$ will be identical to Case 2.

\section{Proof of Main Theorem}
\label{Appdx:proofmaintheorem}
\subsection{Setup and Notation}
For notational convenience let $W=[1]\otimes R_a$. For an irreducible representation $R_\beta \subseteq W$, let $\hat{P}_\beta$ denote the projector onto $R_\beta$ and $d_\beta = \text{dim}(R_\beta) =\Tr(\hat{P}_\beta) $. The dimension of the full space is given by $\text{dim}(W)= dd_a$. We denote the size of the Young tableau $[\lambda_a]_L$ to be $|\lambda_a| \coloneq  \sum_i (\lambda_a)_i $. We denote the $A$-th \textit{generator} of $\SON{d}$ acting on a representation $R$ as the operator $\hat{X}_A^{(R)}$, where $A$ runs over all $d(d-1)/2$ generators of $\SON{d}$. We take the generators to be Hermitian, $(\hat{X}_A^{(R)})^\dagger=\hat{X}_A^{(R)}$, antisymmetric, $(\hat{X}_A^{(R)})^T=-\hat{X}_A^{(R)}$, and satisfy the commutation relation $[\hat{X}_A^{(R)},\hat{X}_B^{(R)}]=if_{ABC} \hat{X}_C^{(R)}$, for real structure constants $f_{ABC}$ \cite{Georgi1999generators}. The quadratic Casimir operator is defined as $\hat{C}_2(R) \coloneq \hat{X}_A^{(R)}\hat{X}_A^{(R)}$, where the repeated index $A$ implies summation. We denote the scalar eigenvalue of $\hat{C}_2(R)$ acting on a specific irreducible representation $R_\beta \subseteq R$ as $c_2(R_\beta)$. Finally, the relevant operator in C2 is $\hat{M}\equiv M_{ia,jb}\coloneq\Gamma_{ibI}\Gamma_{jaI}$. Because $\hat{M}$ is an intertwiner, it commutes with the group action on the full space $W$ and thus belongs to the centralizer algebra for the group action, $\mathrm{End}_{\SON{d}}(W)$. Observe that $W$ is multiplicity free: by the Newell-Littlewood rules, each $R_\beta \subseteq W$ is formed by adding/subtracting a single box to $R_a$ (followed by modification rules), and each resulting tableau appears only once. A tableau with $d/2$ rows in even $d$ further splits into distinct self-dual and anti-self-dual representations, so each irreducible representation appears in $W$ at most once. By Schur's lemma, $\mathrm{End}_{\SON{d}}(W)$ is spanned by projectors onto the irreducible representations $\hat{P}_\beta$, and since $\hat{M}\in \mathrm{End}_{\SON{d}}(W)$,
\begin{equation}
    \hat{M} = \sum_\beta \kappa_I(\beta) \hat{P}_\beta.
\end{equation}
 To characterize $\mathrm{End}_{\SON{d}}(W)$, we study the \textit{interaction generator} $\hat{X}_{\mathrm{int}} = \hat{X}_A^{([1])} \otimes \hat{X}_A^{(R_a)}$ (the repeated $A$ index again implies summation), which is the natural invariant coupling between $R_a$ and $[1]$. Such a generator arises when considering the quadratic Casimir operator $\hat{C}_2$ acting on $W$, which is constructed from the total generator $\hat{X}_A^{(W)} = \hat{X}_A^{([1])}\otimes \mathbb{I} + \mathbb{I}\otimes \hat{X}_A^{(R_a)} $:
\begin{equation}
    \hat{C}_2(W)= \hat{X}_A^{(W)}\hat{X}_A^{(W)}= (\hat{X}_A^{([1])}\otimes \mathbb{I} + \mathbb{I}\otimes \hat{X}_A^{(R_a)} )^2 = \hat{C}_2(R_a)+\hat{C}_2([1]) + 2\hat{X}_{\mathrm{int}}.
\end{equation}
Since Casimirs commute with the group action, Schur's lemma implies that each Casimir operator acts as a scalar on every irreducible subspace. When acting on $R_\beta$, $\hat{C}_2(R_a)$ and $\hat{C}_2([1])$ become their eigenvalues $c_2(R_a)$ and $c_2([1])$, while $\hat{C}_2(W)$ becomes $c_2(R_\beta)$. It follows that the eigenvalues $x_\beta$ of $\hat{X}_{\mathrm{int}}$ associated with the irreducible representation $R_\beta\subseteq W$ are given by
\begin{equation} 
\label{eq:xbetadef}
    x_\beta=\frac{1}{2}\left(c_2(R_\beta)-c_2(R_a)-c_2([1]) \right).
\end{equation}
Since $\hat{X}_{\mathrm{int}}$ is a linear combination of Casimirs, it lies in $\mathrm{End}_{\SON{d}}(W)$, and we can decompose it in terms of its eigenvalues:
\begin{equation}
    \hat{X}_{\mathrm{int}}= \sum_{\beta}x_\beta \hat{P}_\beta.
\end{equation}
 In the next section, we shall prove various properties of the Casimir eigenvalue $c_2(R_\beta)$, including an analytic expression in terms of the Young tableau for $R_\beta$. 

With this setup established, the proof of Theorem~\ref{thm:pform} will proceed as follows. First, we notice that $\hat{M}$ is proportional to the partial transpose of the projector onto $R_I$, denoted $\hat{P}_I$. For convenience, for the remainder of this Appendix, we shall work with a different normalization convention from the main text where this proportionality is an equality. We then exploit properties of the partial transpose and the Casimir eigenvalues to derive a system of equations that must be satisfied if C2 holds. Finally, we show that this system of equations has no valid solutions for non-column irreducible representations, completing the proof.

\subsection{On the Nature of Casimirs}
To complete the proof of Theorem~\ref{thm:pform}, we must find explicit expressions for the Casimir eigenvalues. To do so, we must first introduce the concept of \textit{root systems}. 
\begin{defn}
    A \textbf{root system} $(E,\Phi)$ consists of a finite-dimensional real vector space $E$ with an inner product, and a collection of nonzero vectors $\Phi$ (called roots) that satisfy the following four properties:
    \begin{enumerate}
        \item Vectors in $\Phi$ span $E$,
        \item If $\beta \in \Phi$ and $ c\in \mathbb{R}$, then $c\beta \in \Phi$ only if $c=\pm 1$,
        \item If $\beta,\gamma \in \Phi$ then $s_\beta (\gamma)\in \Phi$, where $s_\beta(\gamma)$ is defined to be 
        \begin{equation}
            s_\beta(\gamma) = \gamma - 2\frac{\langle \gamma ,\beta\rangle}{\langle \beta ,\beta\rangle  }\beta ,
        \end{equation}
        \item For all $\beta,\gamma \in \Phi$, we have that $2\frac{\langle \gamma ,\beta\rangle}{\langle \beta ,\beta\rangle  }$ is an integer. 
        \end{enumerate}
\end{defn}
We further define a \textit{base} as follows:
\begin{defn}
    Let $(E,\Phi)$ be a root system. A subset $\Delta \subseteq \Phi$ is called a base  if it forms a basis for $E$ and every root in $\Phi$ can be expressed as a linear combination of elements in $\Delta$ with non-negative or non-positive integer coefficients. Roots with non-negative (non-positive) coefficients constitute the set of \textbf{positive roots} (\textbf{negative roots}), denoted $\Phi^+$ ($\Phi^{-}$).  
\end{defn}
This framework allows us to uniquely identify representations by their highest weights and compute their Casimir eigenvalues, ultimately enabling the proof of the following theorem:
\begin{thm}
    \label{thm:casimirvalues}
      Let $R$ be an irreducible representation of $\SON{d}$ with a Young tableau $[\lambda]_L$ and weights given by $\lambda=(\lambda_1,\cdots, \lambda_L)$. Then the quadratic Casimir $c_2(R)= \sum_i \lambda_i(\lambda_i+d-2i) \equiv d|\lambda|+g(\lambda)$, where $g(\lambda)\coloneq \sum_i \lambda_i (\lambda_i-2i)$.
\end{thm}
\begin{proof}
    We define the \textit{half-sum of positive roots} as
\begin{equation}
    \delta =(\delta_1,\cdots, \delta_L) \coloneq \frac{1}{2}\sum_{\beta \in \Phi^{+}}\beta.
\end{equation}
According to Proposition 10.6 in Ref.~\cite{Hall2015Further}, the quadratic Casimir eigenvalue for this representation is given by
\begin{equation}
    c_2(R) = \lVert \lambda+\delta    \rVert^2- \lVert \delta \rVert^2= \sum_i \lambda_i(\lambda_i +2\delta_i).
\end{equation}
 The proof is complete if we show that $\delta_i =(d-2i)/2 $. Since the positive roots slightly differ depending on the parity of $d$, we split the analysis into two cases:
    \begin{thmcases}
    \item When $d$ is even, the positive roots take the form $e_j\pm e_k$ for integers $1\leq j<k\leq d/2$. Thus
    \begin{equation}
         2\delta = \sum_{j=1}^{d/2}\sum_{k=j+1}^{d/2} e_j +e_k +e_j-e_k =   2\sum_{j=1}^{d/2}(d/2-j)e_j =\sum e_j (d-2j),
    \end{equation}
which yields the desired result. 
    \item When $d$ is odd, the positive roots are again of the form $e_j\pm e_k$ for integers $1\leq j<k\leq (d-1)/2$, along with the additional roots $e_j$, for $j \in \{1,\cdots, (d-1)/2\}$ . Thus
    \begin{equation}
         2\delta = \sum_{j=1}^{(d-1)/2}\left[e_j +\left(\sum_{k=j+1}^{(d-1)/2} e_j +e_k +e_j-e_k\right)\right] =\sum_j e_j (d-2j),
    \end{equation}
    which matches the even case.
\end{thmcases}
\end{proof}
This analytic form will be crucial for characterizing the properties of $x_\beta$. In fact, we can fully determine $x_\beta$ for any irreducible representation $R_\beta \subseteq W$ by the \emph{content} of its modified box. For a box at row $i$ and column $j$ in the Young tableau $[\lambda]$, its \emph{content} $c(i,j)\coloneq j-i$. We can thus prove the following: 
\begin{lem}
\label{lem:contentxtheorem}
    If $R_\beta$ is obtained from $R_a$ by adding a box of content $c$, then $x_\beta= c$. If $R_\beta $ is obtained from $R_a$ by subtracting a box of content $c$, then $x_\beta = -(d-1+c)$.
\end{lem}
\begin{proof}
    Suppose we add a box at row $r$ (i.e. at position $(r,\lambda_r+1)$). The content of this box is $c=(\lambda_r+1-r)$. Plugging this configuration into Theorem~\ref{thm:casimirvalues}, we find that the Casimir eigenvalue shifts by $c_2(R_\beta) -c_2(R_a) = 2\lambda_r+1+d-2r$. Since $c_2([1]) =d-1$, Eq.~\eqref{eq:xbetadef} simplifies to $x_\beta =\lambda_r+1-r=c $. Likewise, suppose we removed a box at $(r,\lambda_r)$, which has content $c=\lambda_r-r$. This removal causes a Casimir eigenvalue shift of $c_2(R_\beta) -c_2(R_a) =-2\lambda_r+1+2r-d$. It follows that $x_\beta = -\lambda_r+1-d+r= -(d-1+c) $, completing the proof. 
\end{proof}
It will also be useful to express $g(\lambda)$ in terms of the columns of $[\lambda]_L$:
\begin{lem}
\label{lem:columnvalues}
    $g(\lambda_a)=\sum_{j\geq 1} (\lambda_a)_j'(2j-2-(\lambda_a)'_j) $
\end{lem}
\begin{proof}
    \begin{align}
        g(\lambda_a)&= \sum_i (\lambda_a)_i^2 -2i(\lambda_a)_i = \sum_i \left(-2i(\lambda_a)_i+\sum_{j=1}^{(\lambda_a)_i} (2j-1)\right)  =\sum_{(i,j)\in \lambda} 2j-1-2i
        \\&=\sum_j (\lambda_a)'_j(2j-1) -(\lambda_a)_j'((\lambda_a)_j'+1)=\sum_j (\lambda_a)_j'(2j-2-(\lambda_a)'_j).
    \end{align}
\end{proof}

We can further prove that the Casimir eigenvalues in our case will be unique (up to an exception):
\begin{thm}
    Let $R_a$ be an irreducible representation of $\SON{d}$ and let $R_{\beta_1},R_{\beta_2} \subseteq W$ be irreducible representations with $R_{\beta_1}\neq R_{\beta_2}$. Then $c_2(R_{\beta_1})\neq c_2(R_{\beta_2})$, unless $d$ is even and $R_a$ has exactly $d/2-1$ rows. In that case, $W$ contains both $[\lambda_a,1]^{\pm}$ (the self-dual and anti-self-dual irreducible representations, respectively), which share a Casimir eigenvalue.  
    \label{thm:uniquenesscasimir}
\end{thm}
\begin{proof}
    Let the irreducible representation $R_a$ be labeled by its highest weight vector $\lambda = (\lambda_1, \dots, \lambda_L)$, where the components are non-increasing: $\lambda_1 \geq \lambda_2 \geq \dots \geq |\lambda_L| $. All components are non-negative, except when $d$ is even, $L=d/2$, and $R_a$ is anti-self-dual. Let $\delta$ be the half-sum of positive roots, as before. By the explicit construction of the $\SON{d}$ root systems in Theorem~\ref{thm:casimirvalues}, $\delta_i= (d-2i)/2$, so the components of $\delta$ are strictly decreasing and non-negative: $\delta_1 > \delta_2 > \dots > \delta_L \geq 0$, with $\delta_L=0$ only when $L=d/2$ for even $d$. According to Proposition 10.6 in Ref.~\cite{Hall2015Further}, the quadratic Casimir eigenvalue for this representation is given by
\begin{equation}
    c_2(R_{a}) = \lVert \lambda+\delta    \rVert^2- \lVert \delta \rVert^2.
\end{equation}
To derive the eigenvalue $c_2(R_{\beta})$ for the representations in the tensor product, we note that tensoring with the vector representation $[1]$ shifts the highest weight by $\nu$: i.e.  $\lambda_{\beta} = \lambda_a + \nu$. For $\SON{d}$, these shifts are the unit vectors $\nu =\pm e_i$ (where $e_i$ are 1 at index $i$ and zero elsewhere), corresponding to adding or subtracting a box from the Young tableau of $R_a$. For odd $d$, there is an additional zero weight shift $\nu=0$, which occurs when adding a box to the Young tableau yields the same tableau after applying the modification rules. The shifted weight $ \lambda_\beta$ can have a negative last component only if $R_\beta$ is anti-self-dual. For a standard shift $\nu = \pm e_i$,
\begin{equation}
    c_2(R_{\beta}) = \lVert \lambda+\delta  +\nu  \rVert^2- \lVert \delta \rVert^2 = c_2(R_a) + 2\langle \lambda+\delta, \nu\rangle +1.
\end{equation}
Define $\mu_i\coloneq 
\lambda_i +\delta_i$. Because $\lambda$ is non-increasing and $\delta$ is strictly decreasing, the sequence of $\{\mu_i\}$ is strictly decreasing and $\mu_i\leq 0$ is only possible when $i=d/2$ and $d$ is even. Consequently, two distinct shifts $\nu_1\neq \nu_2$ give the same Casimir eigenvalue only if $\langle \mu ,\nu_1\rangle=\langle \mu,\nu_2\rangle$. If both $\nu_1=\pm e_i$ and $\nu_2=\pm e_j$ have the same sign, this requires $\mu_i=\mu_j$ for $i\neq j$, which is impossible, as each component of $\mu$ is distinct. Thus, this can only occur if there exists $\mu_i=-\mu_j$ (possibly with $i=j$). For even $d$, there are two ways this can happen: 
\begin{itemize}
    \item If $L=d/2$ (so $R_a$ is either self-dual or anti-self-dual), then $\mu_L=\lambda_L\neq 0$. Since $|\mu_L|< \mu_j$ for all $j<L$, no two shifts are the same. Note that a single shift changes $\lambda_L$ by $\pm 1$ and since $|\lambda_L| \geq 1 $, it cannot flip the sign of $\lambda_L$. Thus, $W$ never contains both self-dual and anti-self-dual branches. 
    \item If $L=d/2-1$, then adding a box in row $d/2$ corresponds to adding $\nu =\pm e_{d/2}$, with $\mu_{d/2}=0$ (since row $d/2$ of $R_a$ is empty and $\delta_{d/2}=0$). Both choices of $\nu$ yield the same Casimir eigenvalue, corresponding to the exceptional self-dual/anti-self-dual representations $[\lambda_a,1]^{\pm}$. 
\end{itemize}

For odd $d$, all $\mu_i>0$, so the above mechanism cannot occur. The only potential degeneracy arises from a standard shift $\pm e_i$ yielding a net zero shift to the Casimir eigenvalue, causing a clash with the $\nu = 0$ eigenvalue. Setting the shift to zero requires
\begin{equation}
    \pm 2\mu_i + 1 = 0.
\end{equation}
Because $\mu_i$ is strictly positive for odd $d$, the $+e_i$ case can never vanish. For the $-e_i$ case, the above equality requires $\mu_i=\frac{1}{2}$. For odd $d$, the minimal possible value for $\delta_i$ is $\delta_L=\frac{1}{2}$, when $L=(d-1)/2$. Therefore, this zero-shift condition can only be met if $i=L$, $\delta_L = 1/2$, and $\lambda_L = 0$. However, if $\lambda_L=0$, applying the shift $-e_L$ would result in a new highest weight where the $L$-th component is $-1$. Since irreducible representations of $\SON{d}$ for odd $d$ have only non-negative highest weights, this tableau is forbidden. Therefore, the Casimir eigenvalue in the case $\nu=0$ remains unique. Since every valid shift $\nu$ yields a unique eigenvalue, $R_{\beta_1} \neq R_{\beta_2}  \implies c_2(R_{\beta_1} ) \neq c_2(R_{\beta_2} )$ (besides the stated exception), as desired.  
\end{proof}
 The uniqueness immediately provides us the following corollary:
\begin{cor}
\label{cor:uniquevanish}
At most one irreducible representation $R_\beta\subseteq W$ can have $x_\beta=0$.
\end{cor}
\begin{proof}
  By Eq.~\eqref{eq:xbetadef},  $x_\beta =0$ if and only if $c_2(R_\beta)=c_2([1])+c_2(R_a)$. Suppose two distinct irreducible representations $R_{\beta_1}$ and $R_{\beta_2}$ each have $x_{\beta_1}=x_{\beta_2}=0$. Since $R_a$ is fixed, this means that the two representations share a Casimir eigenvalue. By Theorem~\ref{thm:uniquenesscasimir}, this can only hold if $d$ is even, $R_a$ has exactly $d/2-1$ rows, and $\{R_{\beta_1},R_{\beta_2}\}= \{[\lambda_a,1]^{+},[\lambda_a,1]^{-} \}$. We generate $R_\beta$ by adding a single box to the first column of $R_a$, which has column length $d/2 -1$, so the added box has content $c= 1-d/2$. By Lemma~\ref{lem:contentxtheorem}, $x_\beta=c$ for both irreducible representations (since they have the same Casimir eigenvalue). Since we restrict to $d\geq 3$, $c<0$, so $x_\beta\neq 0$, which contradicts our assumption that $x_\beta=0$. Thus, no distinct representations can both have $x_\beta=0$, completing the proof.
\end{proof}
The following will only use Corollary~\ref{cor:uniquevanish}, so the lack of uniqueness of Casimirs in even $d$ does not restrict our main result.

\subsection{Conditions from Partial Transpose Properties}
Let us define the partial transpose $\Theta(\cdot)$ such that $\Theta(M_{ia,jb}) =M_{ib,ja} $. As we noted earlier, $\hat{M}$ and the $R_I$ projector $\hat{P}_I$ are related via  $\Theta(\hat{M}) =\hat{P}_I $. We now establish several general properties of the partial transpose before investigating its action on the generator $\hat{X}_{\mathrm{int}}$:
\begin{lem}
\label{lem:transproperties}
    The partial transpose $\Theta$ satisfies the following three properties for any $Y,Z\in \mathrm{End}(W)$:
    \begin{enumerate}
        \item $\Theta^2=\mathbb{I}$,
        \item  $\Tr(\Theta(Y)Z)= \Tr(Y\Theta(Z))$,
        \item $\Tr(\Theta(Y)^2)=\Tr(Y^2)$.
    \end{enumerate}
\end{lem}
\begin{proof}

    \begin{enumerate}
        \item Fix $Y\in \mathrm{End}(W)$. Applying the definition twice yields $\Theta^2(Y)=Y$.
        \item Fix $ Y,Z\in \mathrm{End}(W)$. Expanding in indices gives
        \begin{equation}
            \Tr(\Theta(Y)Z)= Y_{ib,ja}Z_{jb,ia}=Y_{ia,jb}Z_{ja,ib}= \Tr(Y\Theta(Z)).
        \end{equation}
        \item Fix $ Y\in \mathrm{End}(W)$. Utilizing the first two properties, we find
        \begin{equation}
\Tr(\Theta(Y)^2)=\Tr(\Theta(Y)\Theta(Y))= \Tr(Y\Theta^2(Y))=\Tr(Y^2)
        \end{equation}
    \end{enumerate}
\end{proof}
We now apply this map to the specific generator $\hat{X}_{\mathrm{int}}$:
\begin{lem}
\label{lem:Xproperties}
    The following two properties hold:
    \begin{equation}
        \Theta(\hat{X}_{\mathrm{int}})=-\hat{X}_{\mathrm{int}}, \quad  \Theta(\hat{X}_{\mathrm{int}}^2) =\hat{X}_{\mathrm{int}}^2+(d-2)\hat{X}_{\mathrm{int}}.
    \end{equation}
\end{lem}
\begin{proof}
    The first property follows immediately from the antisymmetry of $\SON{d}$ generators. To prove the second property, we write $\Theta(\hat{X}_{\mathrm{int}}^2)=\hat{X}_A^{([1])}\hat{X}_B^{([1])}\otimes\hat{X}_B^{(R_a)} \hat{X}_A^{(R_a)}  $, which implies that  
    \begin{equation}
    \hat{X}_{\mathrm{int}}^2-\Theta(\hat{X}_{\mathrm{int}}^2) = \hat{X}_A^{([1])}\hat{X}_B^{([1])}\otimes[\hat{X}_A^{(R_a)},\hat{X}_B^{(R_a)}] =\frac{1}{2}[\hat{X}_A^{([1])},\hat{X}_B^{([1])}]\otimes[\hat{X}_A^{(R_a)},\hat{X}_B^{(R_a)}] .
    \end{equation}
    Utilizing the commutation relation $[\hat{X}_A^{(R)}, \hat{X}_B^{(R)}] = i f_{ABC}\hat{X}_C^{(R)}$, where $f_{ABC}$ are the structure constants, the expression becomes: 
    \begin{equation}
        \hat{X}_{\mathrm{int}}^2-\Theta(\hat{X}_{\mathrm{int}}^2)  = -\frac{1}{2} f_{ABC}f_{ABD}\hat{X}_C^{([1])}\otimes \hat{X}_D^{(R_a)} = -\frac{1}{2} c_2([1,1])\delta_{CD} \hat{X}_C^{([1])}\otimes \hat{X}_D^{(R_a)}.
    \end{equation}
    In the final step, we exploited the fact that $-if_{ABC}$ are themselves the generators of the adjoint representation (which corresponds to the rank-2 antisymmetric representation $[1,1]$ of $\SON{d}$, meaning $f_{ABC}f_{ABD} =\Tr(\hat{X}_C^{([1,1])}\hat{X}_D^{([1,1])})=c_2([1,1])\delta_{CD}$; see the derivation of Eq.~\eqref{eq:sixthtrace} for more details). Finally, applying the explicit Casimir formula from Theorem~\ref{thm:casimirvalues} to the adjoint representation yields $c_2([1,1]) = 2(d-2)$. Plugging this back into the expression and rearranging terms gives the desired result. 
\end{proof}
Equipped with these properties, we can derive six key trace identities:
\begin{lem}
    \begin{align}
        \Tr(\hat{M}) &=d_I,\label{eq:firsttrace} \\
        \Tr(\hat{M}\hat{X}_{\mathrm{int}}) &=-x_Id_I, \label{eq:secondtrace}\\
        \Tr(\hat{M}\hat{X}_{\mathrm{int}}^2)&= (x_I^2+(d-2)x_I)d_I, \label{eq:thirdtrace}\\
        \Tr(\hat{M}^2) & =d_I, \label{eq:fourthtrace}\\
    \Tr(\hat{X}_{\mathrm{int}})&=0,\label{eq:fifthtrace}\\
        \Tr(\hat{X}_{\mathrm{int}}^2)&=2c_2(R_a)d_a.\label{eq:sixthtrace}
    \end{align}
\end{lem}
\begin{proof}
\leavevmode\par
\begin{enumerate}
    \item Eq.~\eqref{eq:firsttrace} follows because the partial transpose preserves the trace:
    \begin{equation}
         \Tr(\hat{M}) =\Tr(\Theta(\hat{P}_I))=\Tr(\hat{P}_I)=d_I.
    \end{equation}
    \item Eq.~\eqref{eq:secondtrace} is obtained by combining Lemmas~\ref{lem:transproperties} and \ref{lem:Xproperties}: 
    \begin{equation}
        \Tr(\hat{M}\hat{X}_{\mathrm{int}}) = \Tr(\hat{P}_I\Theta(\hat{X}_{\mathrm{int}}))= -\Tr(\hat{P}_I\hat{X}_{\mathrm{int}})=-x_Id_I.
    \end{equation}
    \item  Eq.~\eqref{eq:thirdtrace} similarly follows from Lemmas~\ref{lem:transproperties} and \ref{lem:Xproperties}: 
    \begin{equation}
        \Tr(\hat{M}\hat{X}_{\mathrm{int}}^2) = \Tr(\hat{P}_I\Theta(\hat{X}_{\mathrm{int}}^2))= \Tr\left(\hat{P}_I(\hat{X}_{\mathrm{int}}^2+(d-2)\hat{X}_{\mathrm{int}}) \right)= (x_I^2+(d-2)x_I)d_I .
    \end{equation}
    \item  Eq.~\eqref{eq:fourthtrace} directly follows from the third property in Lemma~\ref{lem:transproperties}:
    \begin{equation}
        \Tr(\hat{M}^2) = \Tr(\Theta(\hat{P}_I)^2) =\Tr(\hat{P}_I^2)=d_I.
    \end{equation}
    \item Eq.~\eqref{eq:fifthtrace} holds because $\hat{X}_{\mathrm{int}}$ is built from antisymmetric generators, making it inherently traceless.
    \item Proving Eq.~\eqref{eq:sixthtrace} is more involved. Note that, by the definition of a Casimir operator, for an irreducible representation $R_\beta\subseteq W$,
    \begin{equation}
        \Tr_{R_\beta}(X_A^{(R_\beta)}X_B^{(R_\beta)})=\frac{c_2(R_\beta)d_\beta }{\mathrm{dim}(\SON{d})}\delta_{AB}.
    \end{equation}
     This normalization arises from $c_2(\beta)\mathbb{I}_{\beta}= X_A^{(R_\beta)}X_A^{(R_\beta)} $. Likewise, we know that 
     \begin{equation}
         \Tr_{R_\beta}(X_A^{(R_\beta)}X_B^{(R_\beta)})=k_\beta \delta_{AB},
     \end{equation}
    for a constant $k_\beta$. Taking the trace of the former expression and comparing to the latter provides $k_\beta =(c_2(R_\beta)d_\beta) / \dim(\SON{d})$. Given that $\dim(\SON{d})=\frac{1}{2}d(d-1)$ and $\dim([1])=d$, taking the trace of the interaction generator squared yields
    \begin{equation}
        \Tr({\hat{X}_{\mathrm{int}}^2})= k_{a} k_{[1]}\dim(\SON{d}) = \frac{c_2(R_a)d_ad(d-1)}{\dim(\SON{d}) } =2c_2(R_a)d_a.
    \end{equation}
\end{enumerate}
\end{proof}
Finally, we derive a constraint (which we shall call $C^{*}$) that must hold under C2. Denote the shared value $\kappa_I(R_{\tilde{\alpha}})=\kappa$ and the unique value $\kappa_I(R_\alpha)=\kappa+\Delta$. Then the following theorem holds:
\begin{thm}
    Assume C2 holds. Then $x_I,x_\alpha \neq 0$ and 
    \begin{equation}
        \label{eq:solutionsdelta}
        \Delta = -\frac{x_Id_I}{x_\alpha d_\alpha}, \quad \kappa = \frac{d_I(x_I+x_\alpha)}{dd_ax_\alpha} .
    \end{equation}
    Furthermore, the equation (labeled $C^*$) must be true:
    \begin{equation}
    \label{eq:cstareqn}
        C^*: 2c_2(R_a)(x_I+x_\alpha)=dx_Ix_\alpha(x_I+x_\alpha+d-2)
    \end{equation}
\end{thm}
\begin{proof}
    From our choice in $\kappa_I(\cdot)$, we can expand the operator as 
    \begin{equation}
        \hat{M}= \kappa\sum_{\tilde{\alpha}} \hat{P}_{\tilde{\alpha}}+ (\kappa+\Delta)\hat{P}_\alpha.
    \end{equation}
    Substituting this expansion into Eqs.~\eqref{eq:firsttrace} and \eqref{eq:secondtrace}, we obtain
    \begin{equation}
        \kappa \sum_{\tilde{\alpha}}d_{\tilde{\alpha}} +(\kappa+\Delta) d_\alpha =d_I, \quad  \kappa \sum_{\tilde{\alpha}} d_{\tilde{\alpha}}x_{\tilde{\alpha}}+ (\kappa+\Delta)d_\alpha x_\alpha =-x_Id_I.
    \end{equation}
   Using the identities $0=\Tr(\hat{X}_{\mathrm{int}}) =\sum_{\beta} x_\beta d_\beta $ (from Eq.~\eqref{eq:fifthtrace}) and $\sum_{\beta} d_\beta = dd_a $, these relations simplify to
    \begin{equation}
        \kappa dd_a + \Delta d_\alpha =d_I, \quad \Delta d_\alpha x_\alpha = -x_Id_I. \label{eq:intermediatesolutionsdelta}
    \end{equation}
    Observe that $x_\alpha \neq 0$, since $x_\alpha=0$ would imply $x_I=0$,  contradicting Corollary~\ref{cor:uniquevanish}. Next, we show that $x_I\neq 0$. For the sake of contradiction, assume $x_I=0$ (and $x_\alpha \neq 0$). The second expression in Eq.~\eqref{eq:intermediatesolutionsdelta} forces $\Delta =0$, which reduces the first expression to $\kappa dd_a=d_I$ and $\hat{M}=\kappa \mathbb{I}$. Taking the trace of $\hat{M}^2$ and invoking Eq.~\eqref{eq:fourthtrace} yields $\kappa^2 dd_a=d_I$, forcing $\kappa =\kappa^2$, so $\kappa \in \{0,1\}$. However, $\kappa$ cannot vanish (as that would mean $d_I=0$), nor can it be one (as that implies $d_I=dd_a$, which requires $W$ to consist of a single irreducible representation, while $W$ always decomposes into at least three).  Thus, no valid choice of $\kappa$ exists, showing that $x_I\neq 0$. With these trivial cases voided, solving Eq.~\eqref{eq:intermediatesolutionsdelta} yields Eq.~\eqref{eq:solutionsdelta}. 

    To obtain $C^*$, we evaluate the trace: 
    \begin{equation}
        \Tr(\hat{M}\hat{X}_{\mathrm{int}}^2)= \kappa \sum_{\tilde{\alpha}}x_{\tilde{\alpha}}^2d_{\tilde{\alpha}} +(\kappa+\Delta)x_\alpha^2d_\alpha = \kappa \Tr(\hat{X}_{\mathrm{int}}^2)+\Delta x_\alpha^2d_\alpha.
    \end{equation}
    Plugging in Eq.~\eqref{eq:sixthtrace} and equating the result to Eq.~\eqref{eq:thirdtrace} gives
    \begin{equation}
        2\kappa c_2(R_a)d_a +\Delta x_\alpha^2d_\alpha = (x_I^2+(d-2)x_I)d_I.
    \end{equation}
    Since $\Delta x_\alpha^2d_\alpha=-x_Ix_\alpha d_I$ from Eq.~\eqref{eq:intermediatesolutionsdelta}, this simplifies to
    \begin{equation}
         2 c_2(R_a) \frac{d_I(x_I+x_\alpha)}{dx_\alpha} -x_Ix_\alpha d_I = (x_I^2+(d-2)x_I)d_I.
    \end{equation}
    Multiplying both sides by $dx_\alpha$ and canceling $d_I$ produces condition $C^*$.
\end{proof}
We can further specialize condition $C^*$ by incorporating the constraints imposed by C1: 
\begin{lem}
\label{lem:Gdef}
    Let C1 and C2 hold and assume that $d\geq 4((\lambda_a)_1'+1) $. Then, defining $b_R\coloneq (r,s)$, $b_A=(r+1,s)$, and
    \begin{equation}
        C\coloneq c(b_A) = s-r-1\equiv s-(\lambda_a)_s'-1,
    \end{equation}
    it follows that either $(x_I,x_\alpha)=(C,-(C+d))$ or $(x_I,x_\alpha)=(-(C+d),C)$. Under these conditions, $C^*$ becomes equivalent to
    \begin{equation}
        G(d)\coloneq d(|\lambda_a|+C)+g(\lambda_a)+C^2=0,
    \end{equation}
    with $g(\lambda)$ defined as in Theorem~\ref{thm:casimirvalues}. 
\end{lem}
\begin{proof}
    Recall from the main text that the dimensional bound $d\geq 4((\lambda_a)_1'+1)$ restricts $R_\alpha \otimes R_I$ to a regime free of modification rules. As detailed in Section~\ref{sec:RIchoice}, $R_I$ and $R_\alpha$ must be obtained by subtracting/adding a box in the same column of $R_a$. Specifically, the removed box is located at $(r,s)$, with $(\lambda_a)_r=s$ and the added box is located at $(r+1,s)$ with $(\lambda_a)_{r+1}=s-1$. By Lemma~\ref{lem:contentxtheorem}, the removed box possesses an eigenvalue $x=-(d-1+s-r)=-(d+C)$ and the added box has $x=s-r-1 =C$. Because $x_\alpha$ and $x_I$ correspond directly to these added and removed boxes, the first claim is established. To prove the second claim, notice these choices imply $x_\alpha+x_I=-d$ and $x_Ix_\alpha=-C(d+C)$. Plugging this into condition $C^*$ provides
    \begin{equation}
        -2dc_2(R_a) = 2dC(d+C) \iff  g(\lambda_a) +d|\lambda_a| =-C(d+C) \iff G(d)=0.
    \end{equation}
\end{proof}

\subsection{Proof of the Main Theorem}
We establish one final lemma that $G(d)$ is strictly positive if $R_a \neq [1^p]$ for any $p\geq 1$, effectively completing the proof of our main theorem.
\begin{lem}
\label{lem:positivitylem}
Let $R_a$, $R_I$, and $R_\alpha$ be fixed representations satisfying C1 and C2 (as in Lemma~\ref{lem:Gdef}), with an added/removed box pair in column $s$ defining $R_I$ and $R_\alpha$. Let $\lambda_a$ denote the highest weight vector of $R_a$ and, as before,
\begin{equation}
    C\equiv s-(\lambda_a)'_s-1.
\end{equation}
 If $R_a$ is a non-empty, non-column Young tableau, then
    \begin{equation}
        G(d) >0 \text{ for all } d\geq (\lambda_a)_1'+(\lambda_a)_2',
    \end{equation}
    where $G(d)$ is as defined in Lemma~\ref{lem:Gdef}. 
\end{lem}
\begin{proof}
    Denote $B\coloneq (\lambda_a)_1'+(\lambda_a)_2' $. Recalling that $|\lambda_a| = \sum_j (\lambda_a)'_j$, we expand $G(B)$ using Lemma~\ref{lem:columnvalues}: 
    \begin{equation}
    \label{eq:GBeqn}
        G(B) =C(C+B)+\sum_{j\geq 1}(\lambda_a)_j' (B+2j-2-(\lambda_a)_j') .
    \end{equation}
    Since column lengths in a valid Young tableau are monotonically decreasing, the summand is always non-negative. Further, because $[\lambda_a]_L$ is not a single column, $(\lambda_a)_2'\geq 1$. Since the term $C(C+B) $ can potentially become negative, we analyze the behavior across three distinct cases: 
    \begin{thmcases}
    \item ($C\geq 0$): In this case $C(C+B)\geq 0$, allowing us to bound $G(B)$ by the first term in the summand:
    \begin{equation}
        G(B) \geq \sum_{j\geq 1}(\lambda_a)_j' (B+2j-2-(\lambda_a)_j')  \geq (\lambda_a)_1'(\lambda_a)_2'>0. 
    \end{equation}
    \item ($C<0$, $c=1$): The removable box is located at $((\lambda_a)_1',1)$, implying that $C=-(\lambda_a)'_1$. Plugging this in, we obtain $C(C+B)= -(\lambda_a)'_1(\lambda_a)'_2$, which exactly cancels the $j=1$ term in the summand for $G(B)$. This leaves 
    \begin{equation}
        G(B) = \sum_{j\geq 2}(\lambda_a)_j' (B+2j-2-(\lambda_a)_j') \geq (\lambda_a)_2'( (\lambda_a)_1'+ 2) >0.
    \end{equation}
    \item ($C<0$, $c\geq 2$): Since $c\geq 2$, the monotonic decrease of column lengths means $(\lambda_a)_c'\leq (\lambda_a)_2'$, which bounds the values of $-C$ to $1\leq -C\leq (\lambda_a)_2'-1$. Furthermore, the condition  $ (\lambda_a)_2'\leq  (\lambda_a)_1'$ implies $ (\lambda_a)_2'\leq B/2$. Using the property that the function $f(x) = x(B-x)$ monotonically increases on the interval $[0, B/2]$, we find 
    \begin{equation}
        -C(B+C) \leq ( (\lambda_a)_2'-1)(B-(\lambda_a)_2'+1) = ((\lambda_a)_2'-1)((\lambda_a)_1'+1).
    \end{equation}
    Bounding Eq.~\eqref{eq:GBeqn} by the first two terms in the sum alongside the above inequality obtains
    \begin{align}
        G(B)&\geq   (\lambda_a)_1'(\lambda_a)_2'+(\lambda_a)_2'((\lambda_a)_1'+2)-((\lambda_a)_2'-1)((\lambda_a)_1'+1)\nonumber\\
        &= ((\lambda_a)_2'+1)((\lambda_a)_1'+1)>0.
    \end{align}
\end{thmcases}
Thus, across all possible parameter ranges, we have $G(B)>0$. Finally, because $c\geq 1$ and $(\lambda_a)_c'\leq (\lambda_a)_1'$, the value $C$ is bounded below by $C\geq -(\lambda_a)_1'$, meaning $|\lambda_a|+C$ satisfies the inequality
\begin{equation}
    |\lambda_a|+C\geq |\lambda_a|-(\lambda_a)_1' =\sum_{j\geq 2} (\lambda_a)_j'\geq (\lambda_a)_2' \geq 1 .
\end{equation}
Expressing $G(d) = G(B) + (|\lambda_a|+C)(d-B)\geq G(B)>0$, it follows that for all $d \geq B$, $G(d) \geq G(B) > 0$, completing the proof. 
\end{proof}

\begin{proof}[Proof of Theorem~\ref{thm:pform}]
 If C1 and C2 hold, then condition $C^*$ (Eq.~\eqref{eq:cstareqn}) must hold. By Lemma~\ref{lem:Gdef}, condition $C^*$ is equivalent to the constraint $G(d)=0$. However, Lemma~\ref{lem:positivitylem} proves that $G(d) > 0$ for all valid dimensions $d \geq B$ whenever $R_a \neq [1^p]$ for $p \geq 1$. This contradiction implies that the two consistency conditions cannot be satisfied for non-column representations. Because we previously established in Section \ref{Sec:pformframewwork} that standard $p$-form symmetries satisfy both consistency conditions, the proof is complete.
\end{proof}

\end{appendix}

\phantomsection
\section*{References}
\addcontentsline{toc}{section}{References}
\makeatletter
\let\bibsection\relax
\makeatother
\bibliography{refs}
\end{document}